\documentclass[journal]{IEEEtran}
\IEEEoverridecommandlockouts
\usepackage{cite}
\usepackage{algorithm,amsbsy,amsmath,amssymb,bbm,mathrsfs,multirow,amsthm,amsfonts}
\usepackage{array}
\usepackage{graphicx}
\usepackage{graphics}
\usepackage{subfigure}
\usepackage{color}
\usepackage{soul}
\usepackage{bbm}
\usepackage{algorithm}
\usepackage[noend]{algpseudocode}

\algrenewcommand\algorithmicforall{\textbf{foreach}}
\algrenewcommand\algorithmicindent{.8em}
\usepackage{mathtools}
\usepackage{pdfpages}

\usepackage{pifont}
\newcommand{\beq}{\begin{equation}}
\newcommand{\eeq}{\end{equation}}
\newcommand{\bitm}{\begin{itemize}}
\newcommand{\ba}{\begin{array}}
\newcommand{\ea}{\end{array}}
\newcommand{\eitm}{\end{itemize}}
\newcommand{\beqn}{\begin{eqnarray}}
\newcommand{\eeqn}{\end{eqnarray}}
\newcommand{\beqno}{\begin{eqnarray*}}
\newcommand{\eeqno}{\end{eqnarray*}}
\newcommand{\bma}{\begin{displaymath}}
\newcommand{\ema}{\end{displaymath}}
\newcommand{\bnu}{\begin{enumerate}}
\newcommand{\enu}{\end{enumerate}}
\newcommand{\bce}{\begin{center}}
\newcommand{\ece}{\end{center}}
\newcommand{\btb}{\begin{tabular}}
\newcommand{\etb}{\end{tabular}}

\newtheorem{lemma}{Lemma}

\usepackage{graphicx}
\usepackage{textcomp}
\usepackage{xcolor}
\def\BibTeX{{\rm B\kern-.05em{\sc i\kern-.025em b}\kern-.08em
    T\kern-.1667em\lower.7ex\hbox{E}\kern-.125emX}}

\begin{document}

\title{\huge Secure Distributed RIS-MIMO over Double Scattering Channels: Adversarial Attack, Defense, and SER Improvement
}

\author{ 
{Bui Duc Son}, Gaosheng Zhao,~\IEEEmembership{Graduate Student Member,~IEEE}, Trinh Van Chien,\\ and {Dong In Kim},~\IEEEmembership{Life Fellow,~IEEE}
\vspace{-1cm}

 \thanks{This research was supported in part by the MSIT (Ministry of Science and ICT), Korea, under the ICT Creative Consilience program (IITP-2020-0-01821) supervised by the IITP (Institute for ICT Planning \& Evaluation). (Corresponding author: Dong In Kim). An earlier version of this paper has been submitted in part to IEEE International Conference on Communications (ICC) 2026. }  
\thanks{Bui Duc Son, Gaosheng Zhao, and Dong In Kim are with the Department of Electrical and Computer Engineering, Sungkyunkwan University, Suwon 16419, South Korea. Emails: buiducson@skku.edu, gaosheng@skku.edu, and dongin@skku.edu.}
\thanks{ Trinh Van Chien is with the School of Information and Communications Technology, Hanoi University of Science and Technology, Hanoi 100000, Vietnam. Emails:  chientv@soict.hust.edu.vn.}
}

\maketitle
\begin{abstract} 
There has been a growing trend toward leveraging machine learning (ML) and deep learning (DL) techniques to optimize and enhance the performance of wireless communication systems. However, limited attention has been given to the vulnerabilities of these techniques, particularly in the presence of adversarial attacks.
This paper investigates the adversarial attack and defense in distributed multiple reconfigurable intelligent surfaces (RISs)-aided multiple-input multiple-output (MIMO) communication systems-based autoencoder in finite scattering environments. We present the channel propagation model for distributed multiple RIS, including statistical information driven in closed form for the aggregated channel. The symbol error rate (SER) is selected to evaluate the collaborative dynamics between the distributed RISs and MIMO communication in depth. The relationship between the number of RISs and the SER of the proposed system based on an autoencoder, as well as the impact of adversarial attacks on the system's SER, is analyzed in detail. We also propose a defense mechanism based on adversarial training against the considered attacks to enhance the model's robustness. Numerical results indicate that increasing the number of RISs effectively reduces the system's SER but leads to the adversarial attack-based algorithm becoming more destructive in the white-box attack scenario. The proposed defense method demonstrates strong effectiveness by significantly mitigating the attack's impact. It also substantially reduces the system's SER in the absence of an attack compared to the original model. \textcolor{black}{Moreover, we extend the phenomenon to include decoder mobility, demonstrating that the proposed method maintains robustness under Doppler-induced channel variations.}

\end{abstract}

\thispagestyle{empty}
\vspace{-0.5cm}
\section{INTRODUCTION}
\label{sec:introduction}

Recent years have experienced the trend of exploiting machine learning (ML) and deep learning (DL) techniques to optimize and enhance the performance of wireless communication systems. The sixth generation (6G) network, thus, is expected to inherit and further advance the strengths of the previous generation by exploiting ML and DL to fulfill the strict requirements of next-generation networks' applications, such as lower latency, higher throughput, and enhanced reliability compared to previous generations. To achieve the desired quality of service (QoS) and quality of experience (QoE), besides the learning-based methods, numerous cutting-edge technologies are anticipated to be integrated into 6G, including Massive Multiple-Input Multiple-Output (mMIMO), cell-free mMIMO, and millimeter-wave (mmWave) communications \cite{shi2024ris}. However, these advancements must overcome existing hardware limitations and technical challenges. Among these promising technologies, Reconfigurable Intelligent Surfaces (RISs) have appeared as an effective solution to improve the performance of 6G, particularly in the case of non-line-of-sight (NLOS) propagation path between the transmitter and receiver \cite{9475160}. Many studies have shown that RIS-assisted Multiple-Input Multiple-Output (MIMO) explicitly improves communication trustworthiness with modulated signals \cite{tang2020mimo,khoshafa2024ris}.

In wireless communication systems, assuming that all electronic components operate flawlessly, the signal modulation process becomes the primary source of error. To address this issue, end-to-end learning-based communication has emerged as a promising solution. Autoencoder-based ML techniques can enhance the physical layer security by encoding signals before transmission. Beyond learning encoding and decoding strategies, these models can also learn channel characteristics, enabling improved system performance under imperfect channel state information compared to traditional signal processing methods \cite{9360873}. Therefore, securing RIS-assisted MIMO autoencoder systems has become an important research direction, particularly for ensuring both security and reliability in AI-driven wireless communication networks \cite{zhao2024generative}.
\begin{table*}[t]
\caption{Comparative Analysis of Related Works on Adversarial Attacks and Defense Strategies in Wireless Communication Systems}
\centering
\begin{tabular}{lllllllll}
\hline
\multicolumn{1}{|l|}{Reference} & \multicolumn{1}{l|}{Year} & \multicolumn{1}{l|}{MIMO} & \multicolumn{1}{l|}{Fading} & \multicolumn{1}{l|}{Multiple attacks} & \multicolumn{1}{l|}{Defense} & \multicolumn{1}{l|}{\begin{tabular}[c]{@{}l@{}}Defense improved \\ system performance\end{tabular}} & \multicolumn{1}{l|}{\begin{tabular}[c]{@{}l@{}}\textcolor{black}{Mobility} \\ \textcolor{black}{considered}\end{tabular}} & \multicolumn{1}{l|}{\begin{tabular}[c]{@{}l@{}}Multiple \\ distributed RIS\end{tabular}} \\ \hline
\multicolumn{1}{|l|}{\cite{8651357}}         & \multicolumn{1}{l|}{2019} & \multicolumn{1}{l|}{\ding{55}} & \multicolumn{1}{l|}{No} & \multicolumn{1}{l|}{\ding{55}} & \multicolumn{1}{l|}{\ding{55}} & \multicolumn{1}{l|}{\ding{55}} & \multicolumn{1}{l|}{\ding{55}} & \multicolumn{1}{l|}{\ding{55}} \\ \hline
\multicolumn{1}{|l|}{\cite{8449065}}         & \multicolumn{1}{l|}{2019} & \multicolumn{1}{l|}{\ding{55}} & \multicolumn{1}{l|}{No} & \multicolumn{1}{l|}{\ding{55}} & \multicolumn{1}{l|}{\ding{55}} & \multicolumn{1}{l|}{\ding{55}} & \multicolumn{1}{l|}{\ding{55}} & \multicolumn{1}{l|}{\ding{55}} \\ \hline
\multicolumn{1}{|l|}{\cite{9838452}}         & \multicolumn{1}{l|}{2022} & \multicolumn{1}{l|}{\ding{55}} & \multicolumn{1}{l|}{No} & \multicolumn{1}{l|}{\ding{51}} & \multicolumn{1}{l|}{\ding{51}} & \multicolumn{1}{l|}{\ding{55}} & \multicolumn{1}{l|}{\ding{55}} & \multicolumn{1}{l|}{\ding{55}} \\ \hline
\multicolumn{1}{|l|}{\cite{zheng2023designing}} & \multicolumn{1}{l|}{2023} & \multicolumn{1}{l|}{\ding{51}} & \multicolumn{1}{l|}{No} & \multicolumn{1}{l|}{\ding{51}} & \multicolumn{1}{l|}{\ding{55}} & \multicolumn{1}{l|}{\ding{55}} & \multicolumn{1}{l|}{\ding{55}} & \multicolumn{1}{l|}{\ding{55}} \\ \hline
\multicolumn{1}{|l|}{\cite{nan2023physical}} & \multicolumn{1}{l|}{2023} & \multicolumn{1}{l|}{\ding{55}} & \multicolumn{1}{l|}{Infinite} & \multicolumn{1}{l|}{\ding{51}} & \multicolumn{1}{l|}{\ding{51}} & \multicolumn{1}{l|}{\ding{55}} & \multicolumn{1}{l|}{\ding{55}} & \multicolumn{1}{l|}{\ding{55}} \\ \hline
\multicolumn{1}{|l|}{\cite{10138665}}       & \multicolumn{1}{l|}{2023} & \multicolumn{1}{l|}{\ding{55}} & \multicolumn{1}{l|}{No} & \multicolumn{1}{l|}{\ding{51}} & \multicolumn{1}{l|}{\ding{51}} & \multicolumn{1}{l|}{\ding{55}} & \multicolumn{1}{l|}{\ding{55}} & \multicolumn{1}{l|}{\ding{55}} \\ \hline
\multicolumn{1}{|l|}{\cite{10402044}}       & \multicolumn{1}{l|}{2024} & \multicolumn{1}{l|}{\ding{55}} & \multicolumn{1}{l|}{Infinite} & \multicolumn{1}{l|}{\ding{51}} & \multicolumn{1}{l|}{\ding{51}} & \multicolumn{1}{l|}{\ding{55}} & \multicolumn{1}{l|}{\ding{55}} & \multicolumn{1}{l|}{\ding{55}} \\ \hline
\multicolumn{1}{|l|}{\cite{10416752}}       & \multicolumn{1}{l|}{2024} & \multicolumn{1}{l|}{\ding{55} (MISO)} & \multicolumn{1}{l|}{Infinite} & \multicolumn{1}{l|}{\ding{51}} & \multicolumn{1}{l|}{\ding{55}} & \multicolumn{1}{l|}{\ding{55}} & \multicolumn{1}{l|}{\ding{55}} & \multicolumn{1}{l|}{\ding{55}} \\ \hline
\multicolumn{1}{|l|}{\cite{10742079}}       & \multicolumn{1}{l|}{2024} & \multicolumn{1}{l|}{\ding{51}} & \multicolumn{1}{l|}{Finite} & \multicolumn{1}{l|}{\ding{51}} & \multicolumn{1}{l|}{\ding{55}} & \multicolumn{1}{l|}{\ding{55}} & \multicolumn{1}{l|}{\ding{55}} & \multicolumn{1}{l|}{\ding{55} (Only Double RIS)} \\ \hline
\multicolumn{1}{|l|}{\cite{10830521}}       & \multicolumn{1}{l|}{2025} & \multicolumn{1}{l|}{\ding{55}} & \multicolumn{1}{l|}{Infinite} & \multicolumn{1}{l|}{\ding{55}} & \multicolumn{1}{l|}{\ding{51}} & \multicolumn{1}{l|}{\ding{55}} & \multicolumn{1}{l|}{\ding{55}} & \multicolumn{1}{l|}{\ding{55}} \\ \hline
\multicolumn{1}{|l|}{\cite{10436107}}       & \multicolumn{1}{l|}{2025} & \multicolumn{1}{l|}{\ding{55}} & \multicolumn{1}{l|}{Infinite} & \multicolumn{1}{l|}{\ding{51}} & \multicolumn{1}{l|}{\ding{51}} & \multicolumn{1}{l|}{\ding{55}} & \multicolumn{1}{l|}{\ding{55}} & \multicolumn{1}{l|}{\ding{55}} \\ \hline
\multicolumn{1}{|l|}{Ours}                 & \multicolumn{1}{l|}{2025} & \multicolumn{1}{l|}{\ding{51}} & \multicolumn{1}{l|}{Finite \& Infinite} & \multicolumn{1}{l|}{\ding{51}} & \multicolumn{1}{l|}{\ding{51}} & \multicolumn{1}{l|}{\ding{51}} & \multicolumn{1}{l|}{\ding{51}} & \multicolumn{1}{l|}{\ding{51}} \\ \hline
\end{tabular}
\label{comparingstuides}
\end{table*}

\subsection{Related Works}
Numerous studies have investigated optimal Reconfigurable Intelligent Surface (RIS) deployment architectures to maximize the benefits of passive beamforming in wireless networks. Early work focused on the point-to-point scenario, showing that placing a single RIS near either the transmitter or the receiver can effectively overcome blockage and extend coverage range \cite{wu2021intelligent}. For multiuser systems, both centralized and distributed deployment strategies have been proposed: a large, cell-edge RIS can jointly serve multiple users through coherent phase alignment \cite{pan2020multicell}, while numerous smaller RIS panels distributed across the cell improve spatial diversity and reduce end-to-end path loss \cite{zhang2021intelligent}. More recently, multi-reflection architectures have emerged, in which two or more RISs are deployed in series, typically one adjacent to the base station and another near the user, to achieve ultra-high beamforming gains in environments with severe blockages \cite{mei2022intelligent}. These representative deployment scenarios, including single RIS placement, centralized versus distributed multiuser RIS deployments, and multi-reflection RIS architectures, illustrate the key trade-offs between coverage, capacity, and deployment complexity that guide practical RIS integration in emerging 6G MIMO systems.

Beyond deployment strategies, RIS-aided communication systems have been studied extensively, such as transmit power minimization \cite{yaswanth2023robust}, the maximization of the minimum signal-to-interference plus noise ratio (SINR) \cite{papazafeiropoulos2023max}, sum-rate maximization \cite{zhang2022sum}, and the minimization of mean square error (MSE) \cite{choi2024wmmse}. Nevertheless, optimizing the phase shift of RISs is a non-convex problem, posing significant challenges for traditional optimization techniques \cite{zhang2025sum}. Thanks to ML and DL, we now have promising alternative approaches for solving such complex problems. In particular, autoencoders, a specialized class of neural networks, have attracted growing interest among researchers for their potential to optimize RIS-assisted communication systems effectively. For example, the authors in \cite{9745781} proposed an end-to-end learning framework to enhance the communication reliability of RIS-assisted MIMO systems based on deterministic channel realization. In the context of  MIMO-orthogonal frequency-division multiplexing (OFDM) in mm-Wave, the convolutional denoising autoencoder model is exploited for channel prediction by the authors in \cite{10120965}. Additionally, M. H. Wu et al. \cite{wu2025intelligent} proposed a variational denoising autoencoder model with cross-attention precoding (DVAE-CATT-Precoding) to address subcarrier channel interference. The numerical results indicated that DVAE-CATT-Precoding improves MSE, achievable rate, generalizability, and robustness compared to earlier research. Despite these advances, the number of studies related to the implementation of distributed RIS-MIMO systems is currently limited because acquiring CSI for all links is highly challenging and often impractical, and ML models may struggle to learn optimal solutions effectively while a large number of parameters need to be optimized. Notable efforts in this research field are \cite{9977540} and \cite{10136735} for single and double RIS configurations, respectively.
\begin{table}[ht]
\centering
\caption{Summary of Notations}
\label{tab:notations}
\begin{tabular}{|l|l|}
\hline
\textbf{Symbol} & \textbf{Description} \\
\hline
$\mathbf{A}, \mathbf{a}$ & Matrix and vector (bold uppercase and lowercase) \\
$(\cdot)^T$ & Transpose of a matrix or vector \\
$(\cdot)^H$ & Hermitian (conjugate transpose) \\
$(\cdot)^*$ & Conjugate of a complex scalar or vector \\
$\mathbf{B}_N$ & Identity matrix of size $N \times N$ \\
$\text{diag}(x_1, \dots, x_N)$ & Diagonal matrix with elements $x_1$ to $x_N$ \\
$\mathbf{A}^{-1}$ & Inverse of matrix $\mathbf{A}$ \\
$\mathbb{F}$ & Generic field (e.g., $\mathbb{R}$ or $\mathbb{C}$) \\
$\|\cdot\|_2$ & Euclidean norm of a vector \\
$\|\cdot\|_F$ & Frobenius norm of a matrix \\
$\mathbb{E}\{\cdot\}$ & Expectation operator \\
$\otimes$ & Kronecker product \\
$\mathcal{CN}(\cdot, \cdot)$ & Circularly symmetric complex Gaussian distribution \\
$\mathcal{O}(\cdot)$ & Big-$\mathcal{O}$ notation (computational complexity) \\
\hline
\end{tabular}
\end{table}
Nevertheless, although previously mentioned papers focus on enhancing the systems' performance by exploiting ML and DL, they do not examine these models' weaknesses under the effect of adversarial attacks, which are designed to fool ML and DL. To address this concern, the authors in \cite{10460991} investigated the potential for exploiting adversarial attack approaches to degrade the performance of the 6G-IoT systems. Adversarial attack strategies can be classified into three types based on the adversary's level of knowledge about the victim model: black-box (no knowledge), gray-box (limited knowledge), and white-box (full knowledge). Among these, the white-box is the worst scenario and serves as a lower-bound performance for legitimate systems because the adversary finds it difficult to obtain that information in practice. For example, in \cite{8651357}, the authors exploited Additive white Gaussian noise (AWGN) channels and the simple configuration of the autoencoder to examine the destructive abilities of adversarial attacks in end-to-end communication systems as the signal-to-noise ratio (SNR) varies. The numerical results showed that the fast gradient method (FGM) is more effective than a jamming attack in terms of maximizing the target model's bit error rate (BER). To handle the limitation of this work, Son et al. \cite{10742079} proposed a method based on projected gradient descent (PGD), enhanced from FGM, to maximize the symbol error rate (SER) of the proposed autoencoder in infinite scattering environments. The results showed that PGD effectively increases the autoencoder's SER compared to FGM and jamming attacks. In another study, the authors in \cite{9838452} introduced a generative adversarial network (GAN)-based framework method where the adversary is considered as a generative network and the autoencoder-based system is trained to defend against it via a minimax optimization game. However, this paper only used the model in \cite{8651357}, which had only a single antenna for both transceivers and used simple AWGN channels, leading to a lack of practical application. Motivated by this gap, we conduct a comprehensive investigation of the vulnerabilities and defense strategy for distributed RIS-assisted MIMO systems under adversarial attacks. 

\begin{figure}[ht]
\centering
\includegraphics[width=0.9\linewidth]{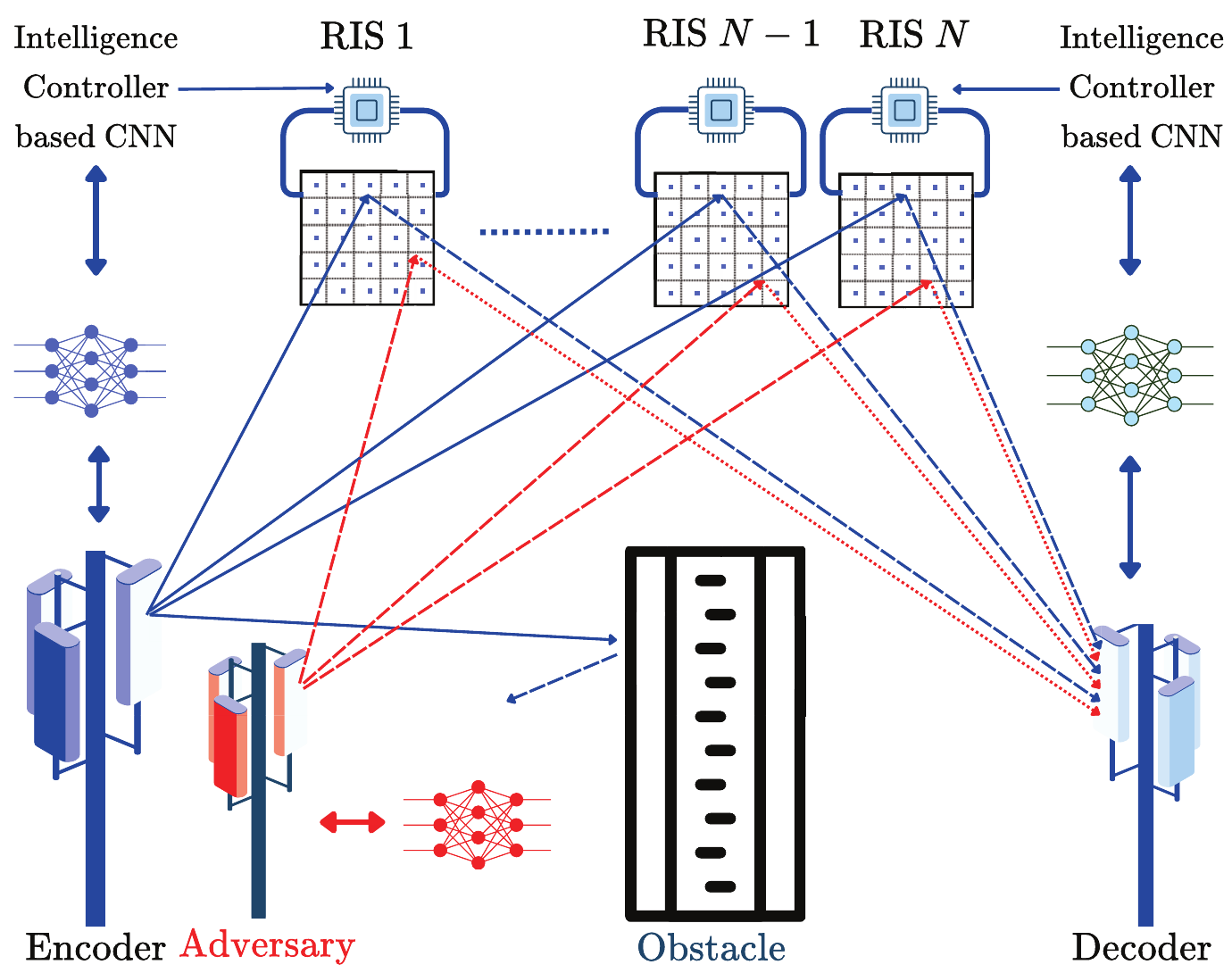}
 \caption{Adversarial attack on distributed RIS-aided MIMO autoencoder.}
  \label{all}
\end{figure}
\subsection{Main Contributions}
This paper investigates the SER of the multiple distributed RIS-assisted MIMO system in a harsh environment by considering the correlation between the system's SER and the number of RISs. Regarding the robustness analysis, we examine the adversarial attack when both the transceiver and adversary have multiple antennas under imperfect channels, instead of a single antenna and AWGN channels \cite{8651357,8449065,10416752,9838452}. Furthermore, we first investigate the novel capability of the adversarial training-based method, which not only mitigates the impact of adversarial attacks but also improves model reliability in terms of the SER. To highlight the research gap and position our contributions, Table~\ref{comparingstuides} summarizes key prior studies on adversarial attacks and defenses in wireless communication and our main contributions are as follows:
\begin{itemize}
    \item \textcolor{black}{We extend the concept of double RIS to a generalized multiple distributed RIS configuration under the double-scattering channel model, in contrast to existing works that mainly consider single or double RIS under AWGN channels or infinite-scattering assumptions. This provides a more realistic framework for analyzing distributed RIS-aided MIMO systems.}


    \item  \textcolor{black}{We analyze the impact of adversarial attacks on distributed RIS-assisted MIMO systems in practical fading environments. Our results highlight that while increasing the number of RISs improves the baseline SER, it simultaneously amplifies the vulnerability to adversarial perturbations, an aspect not captured in prior works under idealized assumptions.}

    \item   \textcolor{black}{To the best of our knowledge, this is the first study to propose an adversarial training-based defense mechanism for distributed RIS-MIMO autoencoders under the finite scattering environment. The method not only mitigates the effect of attacks but also enhances SER performance even in the absence of attacks, showing strong practical value.}

    \item Numerical results show that the quadruple RIS configuration achieves an SER of $10^{-5}$ at 10~[dB], outperforming the double RIS (12~[dB]) and single RIS (18~[dB]) setups. \textcolor{black}{The proposed attack achieves SER degradation for the quadruple-RIS system with reduced complexity compared to \cite{10742079}.} The proposed defense strategy effectively mitigates adversarial attacks and improves the system's reliability. 
    \item \textcolor{black}{We further investigate the system performance under decoder mobility by modeling Doppler-induced channel dynamics. \textcolor{black}{ The results demonstrate that the proposed defense maintains robustness in both static and mobile scenarios, thereby underscoring its practicality for 6G deployments subject to realistic fading and security threats.}}
\end{itemize}


For clarity and ease of reference, the key notations used throughout this paper are summarized in Table~\ref{tab:notations}.
The remainder of this paper is organized as follows: Section~\ref{sec:SYSTEM MODEL} presents the system model, including the double-scattering channel model and the autoencoder architecture. Next, Section~\ref{Universal adversarial attack and defense} introduces the universal adversarial attack and the proposed defense mechanism. After that, Section~\ref{Numerical results} discusses the simulation results. Finally, Section~\ref{conclusion} concludes the paper with a summary of the main findings.

\section{SYSTEM MODEL}
\label{sec:SYSTEM MODEL}
In this section, we examine multiple RIS-Assisted MIMO systems consisting of an encoder, a decoder, and $N$ RISs. The encoder and the decoder are equipped with $K_e$ and $K_d$, while the $m^{th}$ RIS has $N_{m}$ passive reflecting elements. Additionally, the uniform linear array (ULA) configuration is exploited to arrange the antenna arrays of the encoder and decoder, whereas passive reflecting elements are arranged in the uniform planar array (UPA) configuration. The phase shift matrix of $m^{th}$ RIS are formulated as
\begin{equation}
\begin{aligned}
    \mathbf{\Phi}_{m} &= \mathrm{diag}( \beta_{m,1}e^{j\theta_{m,1}}, \ldots,  \beta_{m,N_{i}}e^{j\theta_{m,N_{i}}})\\
                      &= \mathrm{diag}(\mathbf{c}_m^*), \quad \quad m = 1,2,\ldots,N,
\end{aligned} 
\end{equation}
in which, $\mathbf{c}_m^* = (\beta_{m,1}e^{j\theta_{m,1}}, \ldots,  \beta_{m,N_{i}}e^{j\theta_{m,N_{i}}})$; $\theta_{m,i}$ $(\pi \geq \theta_{m,i} \geq -\pi)$ and $\beta_{m,i}$ $(1 \geq \beta_{m,i} \geq 0)$ are the phase and the magnitude of $i^{th}$ passive reflecting element of $m^{th}$ RIS. We assume that one is the unit signal reflection ($\beta_{m,i} = 1,\ \forall m,i$) because of the recent advancements in lossless meta-surface \cite{8811733}. We also consider the harsh propagation scenario in which there is no existing direct link between the encoder and the decoder due to something such as barriers or obstacles blocking it, thereby emphasizing the role and impact of multiple RISs in enhancing the reliability of the considered system. The encoder exploits $M$-Quadrature Amplitude Modulation ($M$-QAM) to modulate the signal $\mathbf{s}$ with the expectation of $\mathbf{ss}^H$ equal to identical matrix $\mathbf{B}_{K_{s}}$, i.e., $\mathbb{E}\{\mathbf{ss}^H \} = \mathbf{B}_{K_{s}}$, where $K_{s}$ is the number of transmitting data streams from the encoder. The active beamforming matrices (linear precoder) $\mathbf{B} \in \mathbb{C}^{K_e \times K_s}$, which satisfy $||\mathbf{B}||^2_F = K_s$, is used to \textcolor{black}{signal} precoding. The received signal at the decoder, denoted by $\mathbf{y} \in \mathbb{C}^{K_d}$, is
\begin{equation}
    \mathbf{y} =\sqrt{\frac{P}{K_s}} \sum\nolimits_{m=1}^{N}(\mathbf{H}_m\mathbf{\Phi}_m\mathbf{D}_m)\mathbf{Bs} + \mathbf{n},
\end{equation}
where $\mathbf{D}_m \in \mathbb{C}^{K_e \times N_m}$ denotes the channel link between the encoder and $m^{th}$ RIS; $\mathbf{H}_m \in \mathbb{C}^{N_m \times K_d}$ is the channel link between the  $m^{th}$ RIS and the decoder; $P$ represents total power required for transmitting \textcolor{black}{signals}; $ \mathbf{n} \sim \mathcal{CN}(0, \sigma^2\mathbf{B}_{K_d})$ is AWGN with \textcolor{black}{mean zero} and variance $\sigma^2\mathbf{B}_{K_d}$. We \textcolor{black}{assume} that there is no channel link between RISs\footnote{As the distance between RISs increases, the power of double-reflection or more (e.g., encoder $\to$ RIS1 $\to$ RIS2 $\to$ decoder) becomes significantly weaker compared to that of single-reflection channels (e.g., encoder $\to$ RIS $\to$ decoder).} to reduce the complexity of the model, and both the encoder and decoder are known for all channel information. The signal received at the decoder is processed further using a combiner matrix as 
\begin{equation}
    \textcolor{black}{\mathbf{z}}=\mathbf{Wy}=\sqrt{\frac{P}{K_s}}\mathbf{WUBs} + \mathbf{Wn},
\end{equation}
where $\mathbf{W} \in \mathbb{C}^{K_s \times K_d}$ is the combiner matrix $(||\mathbf{W}||^2_F = K_s)$ and  $\mathbf{U}$ presents the aggregated channel expressed as
\begin{equation}
    \mathbf{U} =\sum\nolimits_{m=1}^{N} \mathbf{H}_m\mathbf{\Phi}_m\mathbf{D}_m.
    \label{aggregated_channel}
\end{equation}
\begin{figure}[t]
\includegraphics[width=0.9\linewidth]{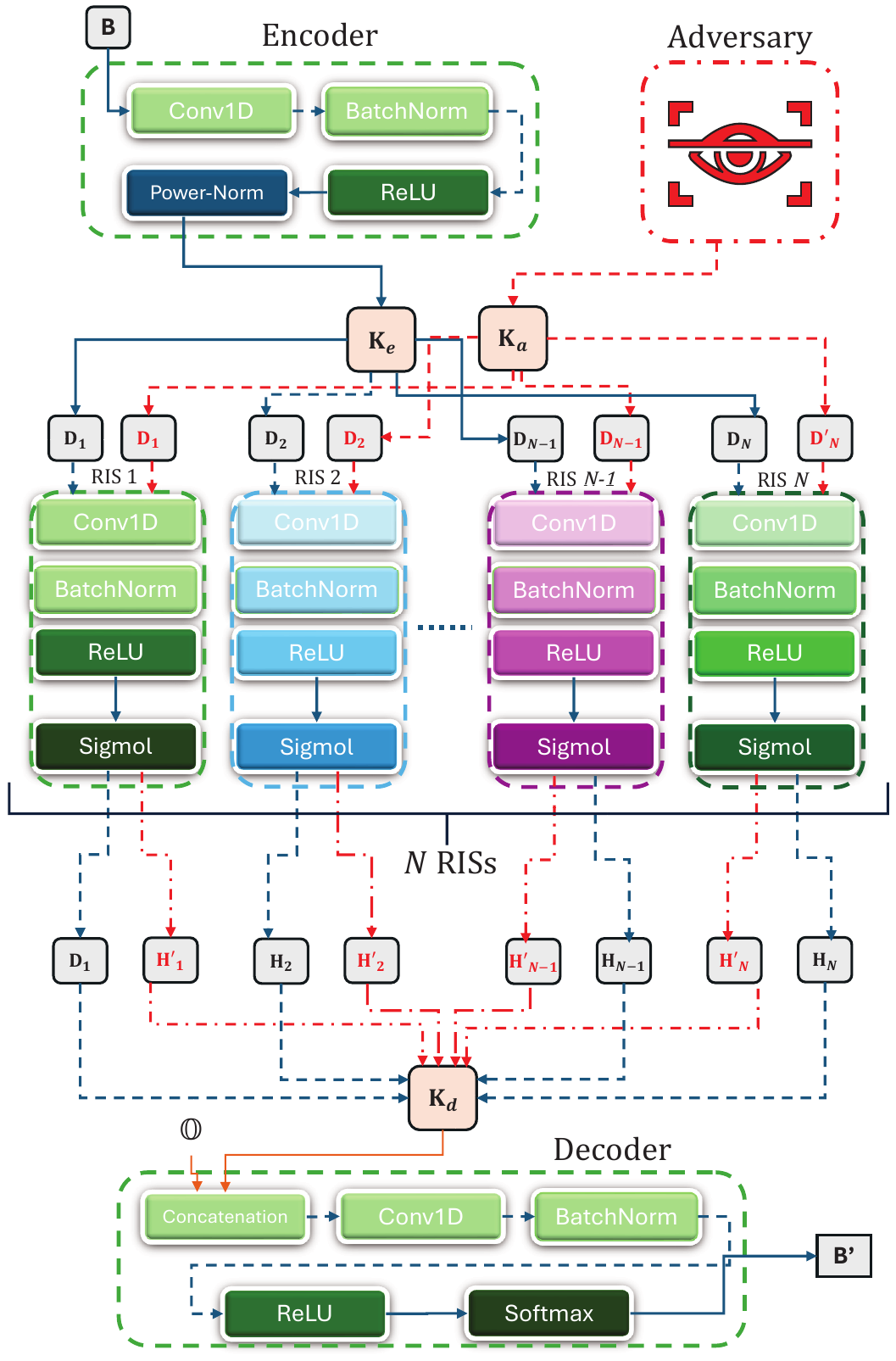}
 \caption{The proposed end-to-end learning framework employs a one-dimensional convolutional neural network (1D-CNN) architecture to model each component of the proposed system.}
  \label{map}
\end{figure}
\vspace{-0.5cm}
\subsection{Double-Scattering Channel Model}
\textcolor{black}{Unlike conventional uncorrelated Rayleigh or Kronecker-based correlated fading models, the double scattering channel model captures three essential features of practical propagation: (i) rank deficiency under poor scattering, (ii) spatial correlation at both transceivers and scatterers, and (iii) sensitivity to the number and structure of scatterers \cite{10136735,9531522}. This channel model spans from rich Rayleigh scattering to keyhole-like channels, making it particularly relevant for RIS-assisted MIMO, where blockages, NLOS conditions, and finite scatterers occur. Adopting this model allows our adversarial attack and defense analysis to reflect realistic vulnerabilities and ensures practical relevance for future 6G deployments\footnote{\textcolor{black}{The impact of the number of scatterers has been analyzed in prior works (e.g., \cite{10136735}, \cite{9531522}). In the low-SNR regime, capacity scales with the received energy rather than with diversity, so a smaller number of scatterers concentrates energy along keyhole-like paths and reduces the SER. In the high-SNR regime, diversity becomes dominant, and a larger number of scatterers is beneficial, whereas a smaller number leads to rank deficiency.}}}. Hence, we now examine the MIMO systems where all the channels are quasi-static in coherence blocks and flat across the bandwidth. The channel links, indicated as $\mathbf{L} \in\{ \mathbf{D}_m,\mathbf{H}_m \}, \forall m = 1,2, \ldots, N$ are modeled using Rician fading, expressed as
\begin{equation}
    \mathbf{L} = \sqrt{\omega}\left( \sqrt{\frac{\upsilon}{\upsilon+1}}\mathbf{\overline{L}} +  \sqrt{\frac{1}{\upsilon+1}}\mathbf{\hat{L}}\right),
    \label{5}
\end{equation}
where $\mathbf{\overline{L}}$ and $\mathbf{\hat{L}}$ denote the deterministic line-of-sight (LOS) and NLOS channels; \textcolor{black}{$\upsilon \in \{\epsilon_m,\delta_m\}$ and $\omega \in \{\alpha_m,\gamma_m\}$ present the Rician factors and the distance-dependent large-scale path-loss coefficients for the channel links $\mathbf{L}$. Specifically, $\epsilon_m$ and $\delta_m$ are the  Rician factors of $\mathbf{D}_m$ and $\mathbf{H}_m$, respectively, while $\alpha_m$ and $\gamma_m$ represent their corresponding path-loss exponents.} Then, the deterministic component $\mathbf{\overline{L}}$ is indicated as the product of UPA and ULA response vectors, defined as
\begin{align}
    &\mathbf{a}_{\mathbf{D}_m}(\theta)=[1, e^{j 2\pi \frac{d_l}{\lambda}\cos\left(\frac{\pi}{2}-\theta\right)},\ldots,e^{j 2\pi (K_e -1) \frac{d_l}{\lambda}\cos\left(\frac{\pi}{2}-\theta\right)}]^T, \label{eq6}\\
    &\mathbf{a}_{\mathbf{H}_m}(\theta)=[1, e^{j 2\pi \frac{d_l}{\lambda}\cos\left(\frac{\pi}{2}-\theta\right)},\ldots,e^{j 2\pi (K_d-1) \frac{d_l}{\lambda}\cos\left(\frac{\pi}{2}-\theta\right)}]^T,
\end{align}
where $\lambda$ is the signal wavelength, $d_l$ denotes the antenna spacing, and $\theta$ represents angle-of-departure (AoD) or angle-of-arrival (AoA), depending on the channel. The index $m = 1, 2, \ldots, N$ corresponds to the $m^{th}$ RIS. We define the array response vectors of UPA along horizontal and vertical axes as
\begin{align}
    &\mathbf{a}_{hm}(\theta) = [1, \ldots, e^{j 2\pi \frac{d_h}{\lambda}(N_{hm}-1)\cos(\frac{\pi}{2}-\theta)}]^T,\\
    &\mathbf{a}_{vm}(\theta,\phi) = [1, \ldots, e^{j 2\pi \frac{d_v}{\lambda}(N_{vm}-1)\cos(\theta)\cos(\frac{\pi}{2}-\phi)}]^T, \label{eq9}
\end{align}
where $d_h$ and $d_v$ represent the distances between two adjacent RIS reflecting elements along the horizontal and vertical axes, respectively; $\theta$ and $\phi$ are the elevation AoA/AoD and azimuth AoA/AoD \cite{mailloux2017phased}; and $N_m = N_{vm} \times N_{hm}$ is the total number of reflecting elements at the $m^{th}$ RIS.  The overall UPA response vector for the $m^{th}$ RIS is then given by
\begin{equation}
    \mathbf{a}_m(\theta, \phi) = \mathbf{a}_{vm}(\theta,\phi) \otimes \mathbf{a}_{hm}(\theta), \quad m = 1,2, \ldots, N.
\end{equation}
Based on \eqref{eq6} - \eqref{eq9}, the LOS components of the channel links can be computed as follows
\begin{align}
    &\mathbf{\overline{D}}_m = \mathbf{a}_m({\theta^A_{Em},\phi^A_{Em}})\mathbf{a}_{\mathbf{D}_m}({\theta^D_{Em}})^T \in \mathbb{C}^{N_m \times K_e},\\
    &\mathbf{\overline{H}}_m = \mathbf{a}_{\mathbf{H}_m}({\theta^A_{Dm}})\mathbf{a}_m({\theta^D_{Dm},\phi^D_{Dm}})^T \in \mathbb{C}^{K_d \times N_m},
    \label{12}
\end{align}
where the superscripts $A$ and $D$ denote the AoA and AoD, respectively, while $\theta$ and $\phi$ represent the elevation and azimuth angles.
Besides, we define $M_1, M_2 \in \{K_e, K_d, N_m\}$, hence, matrix $\mathbf{\hat{L}} \in \mathbb{C}^{M_1\times M_2}$ is the NLOS component of the channels. Let denote $ \mathbf{R}_{e,\mathbf{L}} \in \mathbb{C}^{M_1 \times M_1}, \mathbf{S}_{\mathbf{L}} \in \mathbb{C}^{\mathrm{SC}_\mathbf{L} \times \mathrm{SC}_\mathbf{L}}$, and $\mathbf{R}_{d,\mathbf{L}}\in \mathbb{C}^{M_2 \times M_2}$ are encoder, scatterer, and decoder correlation matrices for channel link $\mathbf{L}$, while $\mathbf{T}_{\mathbf{L}}\in \mathbb{C}^{M_1 \times \mathrm{SC}_{\mathbf{L}}}$, $ \mathbf{E}_{\mathbf{L}} \in \mathbb{C}^{\mathrm{SC}_{\mathbf{L}}\times M_2}$, and $\mathrm{SC}_\mathbf{L}$ are the small-scale fading between the encoder and scatters, small-scale fading between \textcolor{black}{scatters and the decoder}, and the number of scatterers in the propagation channel link in $\textcolor{black}{\mathbf{L}} \in \{\mathbf{D}_m,\mathbf{H}_m\}$. Accordingly, the NLOS component $\mathbf{\hat{L}}$ is computed as
\begin{equation}
    \mathbf{\hat{L}} = \sqrt{\frac{1}{\mathrm{SC}_\mathbf{L}}}\mathbf{R}^{0.5}_{e,\mathbf{L}}\mathbf{T}_{\mathbf{L}}\mathbf{S}^{0.5}_{\mathbf{L}}\mathbf{E}_{\mathbf{L}}\mathbf{R}^{0.5}_{d,\mathbf{L}}. \label{13}
\end{equation}
\textcolor{black}{For clarity, equations (\ref{5})-(\ref{12}) present the deterministic array responses of the ULA/UPA structures, while the NLOS term in (\ref{13}) captures scatterer-induced randomness and spatial correlation. Together, these equations characterize the double-scattering 
channel model.} 

The correlation matrices between the decoder and scatterers \textcolor{black}{between scatters and the decoder} are given by the assumption that the scatterers are arranged in the linear array structure \cite{1175470}. Thus, it can be formulated as
\begin{equation}
    [\mathbf{R}_{e,\mathbf{L}}]_{i,j} = \frac{1}{\mathrm{SC}_{\mathbf{L}}}\sum\nolimits_{t=-q }^{q}\exp(j2\pi d_e(i-j)\sin(v_e)).
    \label{14}
\end{equation}
In \textcolor{black}{the above}, $[\mathbf{R}_{e,\textcolor{black}{\mathbf{L}}}]_{i,j}$ is the $(i;j)^{th}$ element of $[\mathbf{R}_{e,\textcolor{black}{\mathbf{L}}}]$; $v_e = \frac{t\psi_e}{1-\mathrm{SC}_{\textcolor{black}{\mathbf{L}}}}$ and $q = 0.5(\mathrm{SC}_{\textcolor{black}{\mathbf{L}}} -1)$; $d_e$ is the antenna spacing of the encoder or decoder; $\psi_e$ and $\mathrm{SC}_{\textcolor{black}{\mathbf{L}}}$ denote the signals' angular spread and the number of scatterers of channel link $\textcolor{black}{\mathbf{L}\in \{\mathbf{D}_m,\mathbf{H}_m\}}$, respectively. Furthermore, the correlation matrices between the RISs and scatterers align with the spatial correlation matrices of a planar antenna array. The correlation matrices between $m^{th}$ RIS and the scatterers $\mathbf{R}_m \in \{ \mathbf{R}_{\textcolor{black}{e},{\mathbf{D}_m}},\mathbf{R}_{\textcolor{black}{d},{\mathbf{H}_m}}\}$  is formulated as
\begin{equation}
    \mathbf{R}_m = \mathbf{R}_{vm}  \otimes \mathbf{R}_{hm}, 
\end{equation}
where $\mathbf{R}_{hm}$ and $\mathbf{R}_{vm}$ are the correlation matrices of RIS along the horizontal and vertical axes, which are computed as
\begin{align}
    &     [\mathbf{R}_{hm}]_{i,j} = \frac{1}{\mathrm{SC}_{\mathbf{L}}}\sum\nolimits_{t=-q^{'} }^{q^{'}}\exp(j2\pi d_h(i-j)\sin(v_h)),\label{16}\\
    &     [\mathbf{R}_{vm}]_{i,j} = \frac{1}{\mathrm{SC}_{\mathbf{L}}}\sum\nolimits_{t=-q^{'} }^{q^{'}}\exp(j2\pi d_v(i-j)\sin(v_v)). \label{17}
\end{align}
In (\ref{16}) and (\ref{17}), $\mathrm{SC}_{\mathbf{L}}$ is the number of scatterers in any channel link of $\mathbf{L}$; $q' = 0.5(\mathrm{SC}_{\mathbf{L}} -1)$; $v_h = v_v = \frac{t\psi_m}{1-\mathrm{SC}_{\mathbf{A}}}$; and $d_h$ and $d_v$ represent the distances between the elements of the RIS in the horizontal and vertical directions, respectively. Finally, we calculate the correlation matrices among the scatterers as
\begin{equation}
[\mathbf{S}_{\mathbf{L}}]_{i,j} = \frac{1}{\mathrm{SC}_{\mathbf{L}}}\sum\nolimits_{t=-q^{'} }^{q^{'}}\exp(j2\pi d_s(i-j)\sin(v_s)),
\label{18}
\end{equation}
where $v_s = \frac{t\psi_s}{1-\mathrm{SC}_{\mathbf{L}}}$, $\psi_s$ is the angular spread, and $d_s$ is the distance between two scatterers. Based on the double-scattering channel characteristics, the following subsection introduces the autoencoder-based CNN architecture designed to optimize end-to-end communication performance. 
\textcolor{black}{In summary, equations (\ref{14})–(\ref{18}) specify the spatial correlation matrices for the 
encoder, decoder, RIS elements, and scatterers. These formulations describe how angular spreads, antenna spacings, and the finite number of scatterers shape the correlation structure of the channel. By incorporating these correlation effects, the double-scattering model provides a more comprehensive and practical description of RIS-assisted MIMO propagation under realistic scattering environments.}
\subsection{Autoencoder  Architecture}
The goal of the autoencoder is to minimize the SER of the detected signals, which can be expressed as
\begin{equation}
    P_S(\mathbf{\Phi}_m,\mathbf{B},\mathbf{s}) = \frac{1}{M}\sum\nolimits_{i=1}^{M}\sum\nolimits_{j=1, j \neq i }^{M}\Pr[\mathbf{s}_i \to \mathbf{s}_j], 
\end{equation}
where $m \in \{1,\ldots,N\}$ and the term $\Pr[\mathbf{s}_i \to \mathbf{s}_j] = Q\left( \sqrt{||\mathbf{WOB}(\mathbf{s}_i \to \mathbf{s}_j)||^2}/(2\sigma^2)\right)$ indicates the probability of symbol $\mathbf{s}_i$ being mistakenly detected as $\mathbf{s}_j$. Here, $Q(\cdot)$ denotes the tail distribution function of the standard normal distribution, which gives the probability that a standard normal random variable exceeds a given value, which is mathematically defined as $Q(x) = \frac{1}{\sqrt{2\pi}} \int_x^{\infty} e^{-\frac{t^2}{2}} dt.$ Hence, the SER minimization problem is defined as
\begin{align}
    (\mathrm{P_1}): \underset{ \mathbf{\Phi}_m,\mathbf{B},\mathbf{s}}{\text{minimize }} &   P_S(\mathbf{\Phi}_m,\mathbf{B},\mathbf{s})\\
     \text{subject to } &  \|\mathbf{B}\|^2_F = K_s,\\
       & |\Phi_m(n)| = 1, \forall m =1,\ldots,N, n = 1, \dots N_m .
    \label{contrain}
\end{align}
Evidently, (P1) is a non-convex problem due to the strong interdependencies among the involved variables. Furthermore, solving (P1) can be quite complex, particularly when dealing with higher orders of the modulation scheme. Consequently, an autoencoder-based machine learning approach is well-suited for addressing this problem. Typically, a one-dimensional convolutional neural network (1D-CNN) is employed to control the precoding matrices jointly and the phase shifts of the RISs\footnote{\textcolor{black}{The autoencoder is trained using an unsupervised learning approach, leveraging transmitted data and CSI to optimize system parameters, including the phase shift coefficients. These coefficients are intelligently predicted by a CNN module and subsequently forwarded to the RISs for real-time configuration. This methodology enables the RISs to operate in a passive manner, thereby maintaining low power consumption and enhancing their feasibility for practical deployment in next-generation wireless networks.}}. One of the most significant advantages is its ability to adapt to different input lengths even after training, which helps reduce the computational complexity of complex systems, such as those involving multiple RIS-assisted MIMO \cite{jiang2019turbo}.
\subsubsection{Encoder-based CNN Architecture}
As illustrated in Fig.~\ref{map}, the encoder is implemented using the 1D-CNN, effectively replacing all conventional components. The input to the encoder consists of a bit sequence, denoted by $\mathbf{b}_i$, which is converted into one-hot vectors of length $M$, corresponding to the modulation order of the $M$-QAM scheme. The length of the transmitting block is denoted by $L_B $, thus, the input of the neural network is $\mathbf{B}_d= [\mathbf{b}_1, \mathbf{b}_2, \ldots, \mathbf{b}_{L_B-1},\mathbf{b}_{L_B} ] \in \mathbb{C}^{M \times L_B}$. Then, it is processed by various one-dimensional convolutions (Conv1D) layers, each followed by a rectified linear unit (ReLU) and one-dimensional batch normalization (1D-BN). To ensure transmission over the double-scattering fading channels, the output signals are normalized by a custom layer named power normalization. We denote the output of the encoder by
\begin{equation}
    \mathbf{T} = \frac{P\mathbf{T'}}{\sqrt{\mathbb{E}[|\mathbf{T'}|^2]}} \in  \mathbb{C}^{K_e \times L_B},
\end{equation}
where $P$ represents the transmit power and $\mathbf{T'}\in\mathbb{C}^{K_e \times L_B}$ is output of the final 1D-CNN layer. Then, $\mathbf{T}$ is transmitted to RISs for the next phase.
\subsubsection{RISs-based CNN Architecture}

In this work, we assume that the double-reflection links between RISs do not exist due to their \textcolor{black}{product path loss, which is negligible.} The received signals at each RIS, thus, are the input of \textcolor{black}{the RIS model}. For $m^{th}$ RIS, the input is defined as $\mathbf{X}_m =[\mathbf{x}_{m1}, \mathbf{x}_{m2}, \ldots, \mathbf{x}_{mL_B}] \in \mathbb{C}^{N_m \times L_B}$, where $\mathbf{x}_{mi}$ is 
\begin{equation}
    \mathbf{x}_{mi} = \mathbf{D}_m^i\mathbf{t}_i, \quad i= 1, \ldots, L_B; m = 1,\ldots, N_m,
\end{equation}
where the superscript $i$ indicates the channel corresponding to the $i^{th}$ symbol.  The real and imaginary parts of the input are separated and reshaped into a tensor of dimensions $2N_m \times L_B$. This tensor is then subsequently processed through multiple 1D-CNN layers, accompanied by 1D-BN and ReLU activation layers. 
The output of the RIS model is the predicted phase shift matrix. Specifically, the predicted phase shift vector at the \(m^{\text{th}}\) RIS is denoted by $\mathbf{\hat{\Theta}}_m = [ \hat{\mathbf{\theta}}_{m1}, \hat{\mathbf{\theta}}_{m2}, \ldots, \hat{\mathbf{\theta}}_{mL_B}] \in \mathbb{C}^{N_m \times L_B}$, where each $\hat{\mathbf{\theta}}_{mi} = \{\hat{\theta}_{m1}^i, \hat{\theta}_{m2}^i, \ldots, \hat{\theta}_{mN_m}^i\}$ corresponds to the predicted phase shifts for the \(i^{\text{th}}\) symbol. The  reflection matrix of $m^{th}$ RIS is then given by
\begin{equation}
\mathbf{\hat{\Phi}}_m^i = \mathrm{diag}(\exp(j\hat{\theta}_{m1}^i), \exp(j\hat{\theta}_{m2}^i), \ldots, \exp(j\hat{\theta}_{mN_m}^i)).
\end{equation}
\subsubsection{Decoder-based CNN Architecture}

At the decoder, the received signal is calculated by exploiting the predicted phase-shift vector current channels as follows:
\begin{equation}
\begin{aligned}
    \mathbf{z}_i &= \left(\sum\nolimits_{m=1}^N\mathbf{H}_m^i\mathbf{\Phi}_m^i\mathbf{D}_m^i\right)\mathbf{t}_i
                 = \mathbf{O}^i\mathbf{t}_i, \quad i = 1,2,\ldots, L_B.
\end{aligned}
\end{equation}
After the decoder received enough $L_B$ symbols, the set of cascaded channels is denoted by $\mathbb{O} = \{\mathbf{O}^1, \mathbf{O}^2, \ldots, \mathbf{O}^{L_B}\}$ and and the received signal matrix is given by $\mathbf{Z} = [\mathbf{z}_1, \mathbf{z}_2, \ldots, \mathbf{z}_{L_B}] \in \mathbb{C}^{K_d \times L_B}$.  These components form the input data of the decoder, which is a tensor with the shape $K_dL_B + K_dK_eL_B$. The real and imaginary parts are used to construct a new tensor with the shape of $2K_dL_B + 2K_dK_eL_B$. This tensor is then processed through multiple 1D-CNN layers, followed by 1D-BN and ReLU activation layers, similar to the encoder structure.  A softmax activation layer is applied at the output to produce a probability vector over all possible transmitted messages for each symbol, denoted as $\mathbf{\hat{P}} = [\mathbf{\hat{p}}_1, \mathbf{\hat{p}}_2, \ldots, \mathbf{\hat{p}}_{L_B}]$. Finally, the decoded message $\mathbf{B'} = [\mathbf{\hat{b}}_1, \mathbf{\hat{b}}_2, \ldots, \mathbf{\hat{b}}_{L_B}]$ is calculated based on the $\mathbf{\hat{p}}_i$ values with the highest probability.
 \subsubsection{Computational complexity of 1D-CNN}
 According to \cite{kiranyaz20211d}, the computational complexity of a 1D-CNN model can be expressed as 
\(\mathcal{O} \left( \sum_{a=1}^{A} L_B k_a^2 F_{a-1} F_a \right)\),  
where $A$ is the total number of layers, $L_B$ represents the block length, $k_a$ denotes the kernel size of the $a^{th}$ layer, and $F_a$ corresponds to the number of filters in the $a^{th}$ layer. In general, the complexity of a 1D-CNN increases with both the block length and kernel size. To characterize the complexity in terms of MIMO system parameters, we assume that the kernel size, number of layers, and number of filters remain constant. Under these assumptions, the computational complexities for the encoder, decoder, and RIS models are given by  
$\mathcal{O}(ML_BK_e$),  $\mathcal{O}(ML_B(K_e + K_d))$,  and $\mathcal{O}(L_B\sum_{m=1}^NN_m^2)$, respectively. Consequently, the overall computational complexity of the end-to-end system can be evaluated as  \(\mathcal{O}(L_BK_eK_d+L_B\sum_{m=1}^NN_m^2)\). Having established the system model, including the double-scattering channel and the proposed autoencoder-based architecture, we now turn our attention to the vulnerability of such systems under adversarial conditions. In the following section, we introduce universal adversarial attack strategies and corresponding defense mechanisms tailored for RIS-assisted MIMO communication systems.

\section{Universal adversarial attack and defense}

\label{Universal adversarial attack and defense}
In this section, we present Algorithm~\ref{alg:2}. Notably, our proposed algorithm is more comprehensive, as it \textcolor{black}{can be} generalized to support distributed RISs.  which enables an adversary to perform a white-box attack on the proposed distributed RIS-assisted MIMO system by utilizing full knowledge of the model and channel state information. To enhance system reliability, we also propose Algorithm~\ref{alg:training} that incorporates universal perturbations during training to improve the robustness of the communication model against adversarial attacks.
\subsection{Adversarial Attack}
The perturbed received signal at the decoder is computed using the predicted phase-shift vectors corresponding to the current channel realization. Specifically, for each symbol $i$, the received signal under adversarial perturbation is given by
\begin{equation}
\begin{aligned}
    \mathbf{z'}_i &=  \left(\sum\nolimits_{m=1}^N\mathbf{H}_m^i\mathbf{\Phi}_m^i\mathbf{D}_m^i\right)\mathbf{t}_i+\left(\sum\nolimits_{m=1}^N\mathbf{H'}_m^i\mathbf{\Phi}_m^i\mathbf{D'}_m^i\right)\mathbf{u}_{\rm e}
                 \\ &= \mathbf{z}_i + \mathbf{C}^i\mathbf{u}_{\rm e}, \quad i = 1,2,\ldots, L_B.
\end{aligned}
\end{equation}
The term $\sum_{m=1}^N\mathbf{H'}_m^i\mathbf{\Phi}_m^i\mathbf{D'}_m^i$, is denoted as $\mathbf{C}^i$, represents the aggregated channel between the attacker and the decoder across RISs. In this expression, \textcolor{black}{$\mathbf{D'}_m^i \in \mathbb{C}^{N_m \times K_e}$ and $\mathbf{H'}_m^i\in \mathbb{C}^{K_d \times N_m}$} denote the channels between the adversary and the RIS $m$, RIS $m$ and the decoder, which are modeled as the double-scattering fading channel. Similar to the decoder, after collecting $L_B$ symbols of perturbation signals, denoted as $\mathbf{Z'} = [\mathbf{z'}_1, \mathbf{z'}_2, \ldots, \mathbf{z'}_{L_B}] \in \mathbb{C}^{K_d \times L_B}$. These are then combined with the cascaded channel $\mathbb{O}$ to form the input data with a size of $2K_e\times L_B + 2K_e \times K_d \times L_B$, which are fed to 1D-CNN layers, each followed by 1D-BN and ReLU activation function. Last but not least, the output  $\mathbf{\hat{P'}} = [\mathbf{\hat{p'}}_1, \mathbf{\hat{p'}}_2, \ldots, \mathbf{\hat{p'}}_{L_B}]$ is achieved by applying the softmax layer, thus, the decoded message $\mathbf{\hat{B'}} = [\mathbf{\hat{b'}}_1, \mathbf{\hat{b'}}_2, \ldots, \mathbf{\hat{b'}}_{L_B}]$, where $\mathbf{\hat{b'}}_i$ is determined based on the highest probability of $\mathbf{\hat{p'}}_i$. It is essential to note that an effective adversarial attack damages the legitimate system, increasing the SER.
\begin{algorithm}[t]
\caption{\textcolor{black}{Multiple Distributed RIS-Aided MIMO Adversarial Example-Based FGM (MRMAEF)}}
\label{alg:2}
\begin{algorithmic}[1]
\color{black}
  \State \textbf{Inputs:} $\mathcal{F}_{\theta}(.)$, ${p}_{\mathrm{PSR}}$, $\sigma^2$, $\{\mathbf{C}^i\}$.
  \For{$i = 1,\dots,n_{i}$}
   \State $\tilde{\mathbf{u}}_{\mathrm{e}} = \mathbf{C}^i \mathbf{u}_{\mathrm{e}}$;
     $\hat{\mathbf{B}}_r = \mathcal{D}\!\left( \sum_{n=1}^N \mathcal{R}_n(\mathcal{E}(\mathbf{B}_d)) + \mathbf{n} + \tilde{\mathbf{u}}_{\mathrm{e}} \right)$;
      \If{$\hat{\mathbf{B}}_r = \mathbf{B}_d$}
          \State $\tilde{\mathbf{u}}_{\mathrm{a}} \leftarrow$ \textbf{FGM\_Update}$(\mathbf{s}_\mathrm{n})$;
          \State    $\mathbf{u}_{\mathrm{a}}
   = \big((\mathbf{C}^i)^H \mathbf{C}^i\big)^{-1} (\mathbf{C}^i)^H \tilde{\mathbf{u}}_{\mathrm{a}}$;
          \If{$\|\mathbf{u}_{\mathrm{e}} + \mathbf{u}_{\mathrm{a}}\|_2^2 \le p_{\mathrm{PSR}}$}
              \State $\mathbf{u}_{\mathrm{e}} \leftarrow \mathbf{u}_{\mathrm{e}} + \mathbf{u}_{\mathrm{a}}$;
          \Else
              \State $\mathbf{u}_{\mathrm{e}} \leftarrow \sqrt{p_{\mathrm{PSR}}}\,\dfrac{\mathbf{u}_{\mathrm{e}} + \mathbf{u}_{\mathrm{a}}}{\|\mathbf{u}_{\mathrm{e}} + \mathbf{u}_{\mathrm{a}}\|_2}$;
          \EndIf
      \EndIf
  \EndFor
  \State \textbf{Output:} $\mathbf{u}_{\mathrm{e}}$.
\hrule
\State \textbf{Function:} \textsc{FGM\_Update}() 
\State \textbf{Input:} $\mathbf{s}_\mathrm{n} = \sum_{n=1}^N \mathcal{R}_n(\mathcal{E}(\mathbf{B}_d)) + \mathbf{n}$, $ \epsilon_{\mathrm{acc}}$.
\State $\boldsymbol{\epsilon} \leftarrow \mathbf{0}^{M \times 1}$
\For{$i = 1$ to $M$}
    \State $\epsilon_{\min} \leftarrow 0$; $\epsilon_{\max} \leftarrow p_{\max}$; $   \mathbf{s}_{\mathrm{norm}} 
   \leftarrow \frac{\nabla_{\mathbf{s}_{\mathrm{n}}}\mathcal{L}(\mathbf{s}_{\mathrm{n}}, e_i)}
          {\big\|\nabla_{\mathbf{s}_{\mathrm{n}}}\mathcal{L}(\mathbf{s}_{\mathrm{n}}, e_i)\big\|_2};$
    \While{$\epsilon_{\max} - \epsilon_{\min} > \epsilon_{\mathrm{acc}}$}
        \State $\epsilon_{\mathrm{avg}} \leftarrow (\epsilon_{\min} + \epsilon_{\max}) / 2$; $\mathbf{s}_{\mathrm{e}} \leftarrow \mathbf{s}_\mathrm{n} - \epsilon_{\mathrm{avg}}\mathbf{s}_{\mathrm{norm}}$;
        \If{$\mathcal{D}(\mathbf{s}_{\mathrm{e}}) = y_{\mathrm{label}}$}
            \State $\epsilon_{\max} \leftarrow \epsilon_{\mathrm{avg}};$
        \Else
            \State $\epsilon_{\min} \leftarrow \epsilon_{\mathrm{avg}};$
        \EndIf
    \EndWhile
    \State $\boldsymbol{\epsilon}[i] \leftarrow \epsilon_{\max};$
\EndFor
\State $e_{\mathrm{target}} \leftarrow \arg \min \boldsymbol{\epsilon}$; $\epsilon^* \leftarrow \underset{i}{\mathrm{min}} \boldsymbol{\epsilon}$; $\tilde{\mathbf{u}}_{\mathrm{a}} \leftarrow~\epsilon^* \nabla_{\mathbf{s}_\mathrm{n}} \mathcal{L}(\mathbf{s}_\mathrm{n}, e_{\text{target}})/{\|\nabla_{\mathbf{s}_\mathrm{n}} \mathcal{L}(\mathbf{s}_\mathrm{n}, e_{\text{target}})\|_2};$
\State \textbf{Output:} $\tilde{\mathbf{u}}_{\mathrm{a}}$.
\end{algorithmic}
\end{algorithm}

Based on the network architecture illustrated in Fig.~\ref{map}, the adversary introduces carefully crafted perturbations to mislead the decoder. To formally characterize this attack strategy, we mathematically model the legitimate and perturbed signals observed at the decoder as follows:
\begin{align}
&\mathbf{s}_{\rm n} =  \sum\nolimits_{i=1}^N\mathcal{R}_i(\mathcal{E}(\mathbf{B}_d)) + \mathbf{n}, \label{eq:w1}\\
&\mathbf{s}_{\rm e} = \sum\nolimits_{i=1}^N\mathcal{R}_i(\mathcal{E}(\mathbf{B}_d)) + \mathbf{n} +  {\tilde{\mathbf{u}}_{\rm{e}}}. \label{eq:w2}
\end{align}
In \eqref{eq:w1} and \eqref{eq:w2}, the functions  $\mathcal{E}(\cdot)$ and $\mathcal{R}_i(\cdot)$ denote the encoder and the 
$i^{th}$ RIS reflection operations within the autoencoder, respectively. These functions are conditioned on the corresponding channel realizations. The adversarial perturbation ${\tilde{\mathbf{u}}_{\rm{e}}}$ in \eqref{eq:w2} is determined by solving the following optimization problem:
\begin{equation}
\label{constrain}
    \begin{aligned}
         & \underset{\mathbf{u}_{\rm e}}{\mathrm{minimize}} \quad \| \mathbf{u}_{\rm e} \|_2 \\ 
         & \text{subject to} \quad \mathcal{D}(\mathbf{s}) \neq \mathcal{D}({\mathbf{s}_{\rm e})},
    \end{aligned}
\end{equation}
where $\mathcal{D}(\cdot)$ is the decoder; ${\tilde{\mathbf{u}}_{\rm{e}}= \mathbf{C}^{i} \mathbf{u}_{\rm e}}$.
We emphasize that the adversary cannot directly use the solution to problem \eqref{constrain} as an optimal adversarial perturbation to fool the decoder, since it depends on a specific input instance. Although the white-box setting is considered in this work, the attacker lacks precise knowledge of the transmitted  \textcolor{black}{random data} and its timing. Inspired by \cite{8651357} and \cite{10742079}, we chose the FGM to construct the primary attack algorithm to generate universal adversarial examples (UAPs)\footnote{It is important to note that more sophisticated attacks such as the Carlini \& Wagner (C\&W) method \cite{carlini2017towards} are not suitable in this context, as they typically require access to the model’s logits (i.e., pre-softmax outputs). Moreover, C\&W focuses on crafting input-specific perturbations, whereas our interest lies in input-agnostic UAPs.} to examine the vulnerabilities of the proposed system. The main reason is that FGM has lower computational complexity than PGD \cite{10742079} while performing well enough to degrade the autoencoder to find the lower bound of the trustworthiness of the legitimate system's performance. Hence, based on Algorithm 1 in \cite{8651357}, we propose the multiple distributed RIS-aided MIMO adversarial example-based FGM (MRMAEF) algorithm to generate the universal adversarial perturbation vector $(\mathbf{u}_{\rm e})$ to maximize the SER of the proposed system. First, Algorithm~\ref{alg:2}\footnote{This algorithm extends Algorithm 1 in \cite{8651357} and Algorithm 1 in \cite{10742079}.} inputs include the model architecture $\mathcal{F}_{\theta}(.)$, the perturbation power constraint $p_{\rm PSR}$, the variance of Gaussian noise $\sigma^2$, and the aggregated channel between the attacker and the decoder across RISs $\mathbf{C}^i$. We denoted $n_i$ as the number of iterations and set $\mathbf{u}_{\rm e} = \mathbf{0}$. Each iteration, the adversary randomly selects a data sample $\mathbf{B}_d $ from the dataset $\mathcal{B} \triangleq \{ \mathbf{B} _d\}$, where $\mathbf{B}_d = [\mathbf{B}_1, \mathbf{B}_2, \ldots, \mathbf{B}_{T}]\in \mathbb{C}^{M\times L_B \times T}$ include $T$ input matrices of the encoder. Hence, the output of the autoencoder is mathematically formulated as 
\begin{equation}
    \hat{\mathbf{B}}_r = \mathcal{D}\left(\sum\nolimits_{n=1}^N\mathcal{R}_n(\mathcal{E}(\mathbf{B}_d))) + \mathbf{n} + {\tilde{\mathbf{u}}_{\rm{e}}} \right) = \mathcal{D}(\mathbf{s}_{\rm e}).
\end{equation} 

\begin{algorithm}[t]
\caption{\textcolor{black}{Adversarial Training for Multiple RIS-Assisted MIMO Systems (ATMRM)}}
\label{alg:training}
\textcolor{black}{
\begin{algorithmic}[1]
\State \textbf{Input:} $\mathcal{F}_{\theta}(\cdot)$, $\sigma^2$, $\mathbf{O}^i$, dataset, $\mathrm{SNR}_{\text{train}}$, PSR, $\mathbf{u}_e$.
\For{$epoch = 1$ to $n_{\text{e}}$}
    \For{$i = 1$ to $L_B$}
        \State $\mathbf{z}_i \leftarrow \mathbf{z}_i + \mathbf{O}^i \mathbf{u}_{\mathrm{e}}$;
        \State $\theta \leftarrow \theta - \alpha \nabla_{\theta} \mathcal{L}(\mathcal{F}_{\theta}(\mathbf{z}_i), y_i)$;
    \EndFor
\EndFor
\State \textbf{Output:} Optimized model $\mathcal{F}'_{\theta}(\cdot)$.
\end{algorithmic}
}
\end{algorithm}
If the decoded message $\hat{\mathbf{B}}_r$ is identical to $\mathbf{B}_d$, i.e., the perturbation is ineffective, an additional adversarial perturbation $\mathbf{u}_{\mathrm{a}}$ is computed using \textcolor{black}{\textbf{FGM\_Update}. Specifically, a null vector $\boldsymbol{\epsilon}=\mathbf{0}^{M\times 1}$ is first initialized. For each symbol, the normalized gradient direction is calculated as
\begin{equation}
   \mathbf{s}_{\mathrm{norm}} 
   = \frac{\nabla_{\mathbf{s}_{\mathrm{n}}}\mathcal{L}(\mathbf{s}_{\mathrm{n}}, e_i)}
          {\big\|\nabla_{\mathbf{s}_{\mathrm{n}}}\mathcal{L}(\mathbf{s}_{\mathrm{n}}, e_i)\big\|_2}.
   \label{eq:snorm}
\end{equation}
Then, given desired perturbation accuracy $\epsilon_{acc}$, a bisection search is performed over $\epsilon\in[0,p_{\max}]$ to determine the minimum $\epsilon$ such that
$\mathcal{D}(\mathbf{s}_{\mathrm{e}})\!\neq y_{\mathrm{label}}$, where $\mathbf{s_{\rm e}} = \mathbf{s_{\rm e}} - \epsilon \mathbf{s}_{\mathrm{norm}}$ and $p_{\max}$ is maximum allowed perturbation norm. After iterating over all $M$ candidate symbols, $\boldsymbol{\epsilon}$ is collected, hence, the most vulnerable target and its minimal perturbation accuracy are determined as $e_{\mathrm{target}}=\arg\min\boldsymbol{\epsilon}$ and $ \epsilon^\star=\underset{i}{\mathrm{min}} \boldsymbol{\epsilon}.$
In the next step, the raw additional adversarial perturbation is computed as follows:
\begin{equation}
   \tilde{\mathbf{u}}_{\mathrm{a}}
   = \epsilon^\star  
   \frac{\nabla_{\mathbf{s}_{\mathrm{n}}}\mathcal{L}(\mathbf{s}_{\mathrm{n}}, e_{\mathrm{target}})}
        {\big\|\nabla_{\mathbf{s}_{\mathrm{n}}}\mathcal{L}(\mathbf{s}_{\mathrm{n}}, e_{\mathrm{target}})\big\|_2}.
   \label{eq:urecv}
\end{equation}
Subsequently, additional adversarial perturbation vector $\mathbf{u}_{\mathrm{a}}$ is obtained as
\begin{equation}
   \mathbf{u}_{\mathrm{a}}
   = \big((\mathbf{C}^i)^H \mathbf{C}^i\big)^{-1} (\mathbf{C}^i)^H \tilde{\mathbf{u}}_{\mathrm{a}}.
   \label{eq:uemitter}
\end{equation}
If the power of the updated perturbation satisfies the perturbation-to-signal ratio (PSR) constraint, i.e.,
$~\|\mathbf{u}_{\mathrm{e}} +~\mathbf{u}_{\mathrm{a}}\|_2^2 \le~p_{\mathrm{PSR}}$, $\mathbf{u}_{\mathrm{e}}$ is updated as follows:
\begin{equation}
\mathbf{u}_{\mathrm{e}} =\mathbf{u}_{\mathrm{e}} + \mathbf{u}_{\mathrm{a}}.
\end{equation}
Otherwise, it is normalized to satisfy the constraint:
\begin{equation}
\mathbf{u}_{\mathrm{e}} =\sqrt{p_{\mathrm{PSR}}}\,
\frac{\mathbf{u}_{\mathrm{e}} + \mathbf{u}_{\mathrm{a}}}
     {\|\mathbf{u}_{\mathrm{e}} + \mathbf{u}_{\mathrm{a}}\|_2}.
\end{equation}}

In summary, to adapt the perturbation to the input structure, a null vector $\pmb{\epsilon} \in \mathbb{C}^{M\times 1}$ is initialized. For each iteration, the algorithm assigns a random label from the dataset and performs a bisection search to find the minimum $\epsilon$ that successfully misleads the decoder using FGM.
After $\epsilon^*$ is determined, $\mathbf{u}_{\rm a}$ is evaluated by computing the gradient of the binary cross-entropy loss (BCE) function $\mathcal{L}(\cdot, \cdot)$ at the target point.
By considering $\mathbf{u}_{\rm a}$ as the additive noise vector and checking the $\ell_2$-norm of total noise vector, $\mathbf{u}_{\rm e}$ is updated as  $\mathbf{u}_{\rm e} \leftarrow \mathbf{u}_{\rm e} + \mathbf{u}_{\rm a}$. The procedure is repeated until all $n_i$ iterations are completed, after which the final universal adversarial perturbation vector $\mathbf{u}_{\rm{e}}$ is obtained. The resulting vector can \textcolor{black}{be generalized} across various inputs, making it a potent input-agnostic attack. The overall computational complexity of the MRMAEF algorithm is $\mathcal{O}(N_iM^2L_B^2)$.  The theoretical foundation of the proposed algorithm is presented in Lemma~\ref{theorem1} and Corollaries 1 and 2.

\begin{lemma}
\label{theorem1}
\textbf{Gradient-based Adversarial Vulnerability of RIS-Aided MIMO Autoencoders:} Let $\mathcal{D}(\cdot)$ denote the trained decoder of a distributed RIS-aided MIMO autoencoder. Let $\mathbf{s}_\mathrm{n} \in \mathbb{C}^{K_d \times L_B}$ be a clean received signal and $y$ the corresponding true label. For any target label $y_{\text{target}} \neq y$, there exists a perturbation vector $\tilde{\mathbf{u}}_{\mathrm{e}}$ with the energy bounded by $\epsilon$, such that
           $ \mathcal{D}(\mathbf{s}_\mathrm{n} + \tilde{\mathbf{u}}_{\rm e}) = y_{\mathrm{target}},$
        provided that $\epsilon$ satisfies
        \begin{equation}
            \epsilon \geq \frac{\delta}{\left\| \nabla_{\mathbf{s}_\mathrm{n}} \mathcal{L}(\mathcal{D}(\mathbf{s}_\mathrm{n}), y_{\text{target}}) \right\|_2},
        \end{equation}
        for some $\delta > 0$, where $\mathcal{L}(\cdot, \cdot)$ is BCE loss function. The optimal perturbation that maximizes the misclassification probability is given by the gradient-based method \cite{sutton2024adversarial}:
        \begin{equation}
\tilde{\mathbf{u}}_{\mathrm{e}} = \epsilon  \frac{\nabla_{\mathbf{s}_\mathrm{n}} \mathcal{L}(\mathcal{D}(\mathbf{s}_\mathrm{n}), y_{\text{target}})}{\left\| \nabla_{\mathbf{s}_\mathrm{n}} \mathcal{L}(\mathcal{D}(\mathbf{s}_\mathrm{n}), y_{\text{target}}) \right\|_2}.
            \label{34}
        \end{equation}
\end{lemma}
\begin{proof}
    Let the decoder output under perturbation be
    \begin{equation}
        \mathcal{D}(\mathbf{s}_\mathrm{n} + \tilde{\mathbf{u}}_{\mathrm{e}}) \approx \mathcal{D}(\mathbf{s}_\mathrm{n}) + J_{\mathcal{D}} \tilde{\mathbf{u}}_{\mathrm{e}},
    \end{equation} 
    where $J_{\mathcal{D}}$ is the Jacobian of the decoder. We perform a first-order Taylor expansion of the loss function around $\mathbf{s}_\mathrm{n}$:
    \begin{equation}
        \mathcal{L}(\mathcal{D}(\mathbf{s}_\mathrm{n} + \tilde{\mathbf{u}}_{\mathrm{e}}), y_{\text{target}}) \approx \mathcal{L}(\mathcal{D}(\mathbf{s}_\mathrm{n}), y_{\text{target}}) + \nabla_{\mathbf{s}_\mathrm{n}} \mathcal{L}^{\top} \tilde{\mathbf{u}}_{\mathrm{e}}.
    \end{equation}
    To maximize the increase in loss, we align $\tilde{\mathbf{u}}_{\mathrm{e}}$ with the gradient direction as illustrated in~(\ref{34}).
    Then,
    \begin{equation}
        \mathcal{L}(\mathcal{D}(\mathbf{s}_\mathrm{n} + \tilde{\mathbf{u}}_{\mathrm{e}}), y_{\text{target}}) \geq \mathcal{L}(\mathcal{D}(\mathbf{s}_\mathrm{n}), y_{\text{target}}) + \epsilon \|\nabla_{\mathbf{s}_\mathrm{n}} \mathcal{L}\|_2.
    \end{equation}
    Thus, if $\epsilon \|\nabla_{\mathbf{s}_\mathrm{n}} \mathcal{L}\|_2 \geq \delta$ for some threshold $\delta$, the decoder prediction can flip, completing the proof.
\end{proof}
\textbf{Corollary 1} (Linear Approximation of the Decoder): \textit{Let the end-to-end system be represented \textcolor{black}{by} a composite mapping: 
\begin{equation}
    \mathcal{F}: \mathbf{s}_\mathrm{n} \mapsto \hat{\mathbf{B}}_r = \mathcal{D} \left( \sum\nolimits_{n=1}^{N} \mathcal{R}_n(\mathcal{E}(\mathbf{B}_d)) + \mathbf{n} \right),
\end{equation} where $\mathbf{s}_\mathrm{n}$ denotes the received signal prior to decoding. Under adversarial attack, a perturbation $\tilde{\mathbf{u}}_e$ is injected, resulting in a new received signal
$\mathbf{s}_e = \mathbf{s}_\mathrm{n} + \tilde{\mathbf{u}}_e$.
Assuming $\mathcal{F}$ is differentiable with respect to $\mathbf{s}_\mathrm{n}$, we can apply a first-order Taylor expansion:
\begin{equation}
    \mathcal{F}(\mathbf{s}_e) \approx \mathcal{F}(\mathbf{s}_\mathrm{n}) + \nabla_{\mathbf{s}_\mathrm{n}} \mathcal{F}(\mathbf{s}_\mathrm{n})^\top \tilde{\mathbf{u}}_e.
\end{equation} 
Due to the high dimensionality of $\mathbf{s}_\mathrm{n}$ in wireless systems, even small perturbations $\tilde{\mathbf{u}}_e$ can result in large changes in $\mathcal{F}(\mathbf{s}_\mathrm{e})$, causing incorrect symbol decisions \cite{goodfellow2014explaining}.}

\textbf{Corollary 2} (Bounded PSR Constraint): \textit{If the perturbation energy is constrained by a PSR budget $\|\tilde{\mathbf{u}}_{\mathrm{e}}\|_2^2 \leq p_{\text{PSR}}$, then the maximum increase in loss satisfies:
\begin{equation}
    \mathcal{L}(\mathcal{D}(\mathbf{s}_\mathrm{n} + \tilde{\mathbf{u}}_{\mathrm{e}}), y_{\text{target}}) \leq \mathcal{L}(\mathcal{D}(\mathbf{s}_\mathrm{n}), y_{\text{target}}) + \sqrt{p_{\text{PSR}}}  \left\|\nabla_{\mathbf{s}_\mathrm{n}} \mathcal{L}\right\|_2.
\end{equation}
Therefore, a successful attack requires:
\begin{equation}
    \sqrt{p_{\text{PSR}}}  \left\| \nabla_{\mathbf{s}_\mathrm{n}} \mathcal{L} \right\|_2 \geq \delta.
\end{equation}} 

Lemma~\ref{theorem1} establishes the minimum perturbation magnitude measured by its $\ell_2$-norm, which is required to flip the decoder’s decision in a RIS-assisted MIMO autoencoder, and shows that the most damaging perturbation aligns with the gradient of the decoder’s loss function. Corollary~1 employs a first-order (linear) approximation to explain why even a small, gradient-aligned perturbation can induce a large change in the decoder’s output in high-dimensional signal spaces. Corollary~2 then imposes a practical power constraint (PSR budget) on the perturbation, demonstrating how this bound caps the attacker’s ability to increase the decoding loss. Collectively, these results confirm that gradient-based attacks (e.g., FGM) are near-optimal in high dimensions and accurately predict the sharp SER degradations observed in our numerical evaluations, including the emergence of small yet highly effective “universal” perturbations across multiple input samples. Having characterized and implemented potent gradient-based adversarial attacks, we now explore \textcolor{black}{their} corresponding defense mechanisms. In the next subsection, we introduce our adversarial training-based algorithm, designed to strengthen the autoencoder against the considered attack method.

\begin{table*}[ht]
\caption{Summary of channel parameters for the secured model used in the simulation}
\label{locationRIS}
\centering
\scriptsize
\renewcommand{\arraystretch}{1.3}
\begin{tabular}{|>{\centering\arraybackslash}p{0.8cm}|>{\centering\arraybackslash}p{2.8cm}|>{\centering\arraybackslash}p{3.5cm}|>{\centering\arraybackslash}p{4.5cm}|>{\centering\arraybackslash}p{3.9cm}|}
\hline
\textbf{Link} & \textbf{Distance} & \textbf{Angle of Arrival (AoA)} & \textbf{Angle of Departure (AoD)} & \textbf{LOS Component} \\ \hline

$\mathbf{D}_1$ & $d_2 - d_1$ & $\theta^A_{E1} = \frac{\pi}{4}$, $\phi^A_{E1} = 0$ & $\theta^D_{E1} = \frac{\pi}{2}$ & \multirow[c]{4}{=}{%
$\bar{\mathbf{D}}_m = \mathbf{a}_m(\theta^A_{Em}, \phi^A_{Em}) \mathbf{a}_E(\theta^D_{Em})$} \\ \cline{1-4}
$\mathbf{D}_2$ & $d_1$ & $\theta^A_{E2} = \frac{\pi}{4}$, $\phi^A_{E2} = 0$ & $\theta^D_{E2} = \frac{\pi}{2}$ & \\ \cline{1-4}
$\mathbf{D}_3$ & $\sqrt{2d_1^2}$ & $\theta^A_{E3} = \frac{\pi}{2}$, $\phi^A_{E3} = 0$ & $\theta^D_{E3} = \frac{\pi}{4}$ & \\ \cline{1-4}
$\mathbf{D}_4$ & $\sqrt{(d_2 - d_1)^2 + d_1^2}$ & $\theta^A_{E4} = \frac{\pi}{4} + \arctan\left(\frac{d_2 - d_1}{d_1}\right)$, $\phi^A_{E4} = 0$ & $\theta^D_{E4} = \arctan\left(\frac{d_1}{d_2 - d_1}\right)$ & \\ \hline

$\mathbf{H}_1$ & $\sqrt{(d_2 - d_1)^2 + d_1^2 + d_3^2}$ & $\theta^A_{R1} = \arctan\left(\frac{d_2 - d_1}{d_1}\right)$ & $\theta^D_{R1} = \frac{\pi}{4} + \arctan\left(\frac{d_2 - d_1}{d_1}\right)$, $\phi^A_{R1} = \arctan\left(\frac{d_3}{d_1}\right)$ & \multirow[c]{4}{=}{%
$\bar{\mathbf{H}}_m = \mathbf{a}_R(\theta^A_{Rm}) \mathbf{a}_m(\theta^D_{Rm}, \phi^D_{Rm})$} \\ \cline{1-4}
$\mathbf{H}_2$ & $\sqrt{2d_1^2 + d_3^2}$ & $\theta^A_{R2} = \frac{\pi}{4}$ & $\theta^D_{R2} = \frac{\pi}{2}$, $\phi^A_{R2} = \arctan\left(\frac{d_3}{\sqrt{2}d_1}\right)$ & \\ \cline{1-4}
$\mathbf{H}_3$ & $\sqrt{d_1^2 + d_3^2}$ & $\theta^A_{R3} = \frac{\pi}{2}$ & $\theta^D_{R3} = \frac{\pi}{4}$, $\phi^A_{R3} = \arctan\left(\frac{d_3}{d_1}\right)$ & \\ \cline{1-4}
$\mathbf{H}_4$ & $\sqrt{(d_2 - d_1)^2 + d_3^2}$ & $\theta^A_{R4} = \frac{\pi}{2}$ & $\theta^D_{R4} = \frac{\pi}{4}$, $\phi^A_{R4} = \arctan\left(\frac{d_3}{d_1}\right)$ & \\ \hline

\end{tabular}
\end{table*}

\begin{table*}[ht]
\caption{Summary of adversarial channel parameters used in simulation scenarios}
\label{locationadverasry}
\centering
\scriptsize
\renewcommand{\arraystretch}{1.3}
\begin{tabular}{|>{\centering\arraybackslash}p{0.6cm}|>{\centering\arraybackslash}p{2.8cm}|>{\centering\arraybackslash}p{3.7cm}|>{\centering\arraybackslash}p{4.5cm}|>{\centering\arraybackslash}p{3.9cm}|}
\hline
\textbf{Link} & \textbf{Distance} & \textbf{Angle of Arrival (AoA)} & \textbf{Angle of Departure (AoD)} & \textbf{LOS Component} \\ \hline

$\mathbf{D}_1'$ & $\sqrt{(d_2 - d_1)^2 + d_3^2}$ & $\theta^A_{E1} = \frac{\pi}{4}$, $\phi^A_{E1} = \arctan(\frac{d_3}{d_2-d_1})$ & $\theta^D_{E1} = \frac{\pi}{2}$ & \multirow[c]{4}{=}{%
$\bar{\mathbf{D}}'_m = \mathbf{a}_m(\theta^A_{Em}, \phi^A_{Em})\mathbf{a}_E(\theta^D_{Em})$} \\ \cline{1-4}
$\mathbf{D}_2'$ & $\sqrt{d_1^2 + d_3^2}$ & $\theta^A_{E2} = \frac{\pi}{4}$, $\phi^A_{E2} = \arctan(\frac{d_3}{d_1})$ & $\theta^D_{E2} =\arctan(\frac{d_3}{\sqrt{2}d_1})$ & \\ \cline{1-4}
$\mathbf{D}_3'$ & $\sqrt{2d_1^2 + d_3^2}$ & $\theta^A_{E3} = \frac{\pi}{2}$, $\phi^A_{E3} = \arctan(\frac{d_3}{\sqrt{2}d_1})$ & $\theta^D_{E3} = \frac{\pi}{4}$ & \\ \cline{1-4}
$\mathbf{D}_4'$ & $\sqrt{(d_2-d_1)^2 + d_1^2 + d_3^2}$ & $\theta^A_{E4} = \frac{\pi}{4} + \arctan\left(\frac{d_2-d_1}{d_1}\right)$, $\phi^A_{E4} = \arctan(\frac{d_3}{\sqrt{d_1^2+(d_2-d_1)^2}})$ & $\theta^D_{E4} = \arctan\left(\frac{d_1}{d_2-d_1}\right)$ & \\ \hline

$\mathbf{H}_1'$ & $\sqrt{(d_2 - d_1)^2 + d_1^2 + d_3^2}$ & $\theta^A_{R1} = \arctan\left(\frac{d_2-d_1}{d_1}\right)$ & $\theta^D_{R1} = \frac{\pi}{4} + \arctan\left(\frac{d_2-d_1}{d_1}\right)$, $\phi^A_{R1} = \arctan\left(\frac{d_3}{d_1}\right)$ & \multirow[c]{4}{=}{%
$\bar{\mathbf{H}}'_m = \mathbf{a}_R(\theta^A_{Rm})\mathbf{a}_m(\theta^D_{Rm}, \phi^D_{Rm})$} \\ \cline{1-4}
$\mathbf{H}_2'$ & $\sqrt{2d_1^2 + d_3^2}$ & $\theta^A_{R2} = \frac{\pi}{4}$ & $\theta^D_{R2} = \frac{\pi}{2}$, $\phi^A_{R2} = \arctan\left(\frac{d_3}{\sqrt{2}d_1}\right)$ & \\ \cline{1-4}
$\mathbf{H}_3'$ & $\sqrt{d_1^2 + d_3^2}$ & $\theta^A_{R3} = \frac{\pi}{2}$ & $\theta^D_{R3} = \frac{\pi}{4}$, $\phi^A_{R3} = \arctan\left(\frac{d_3}{d_1}\right)$ & \\ \cline{1-4}
$\mathbf{H}_4'$ & $\sqrt{(d_2 - d_1)^2 + d_3^2}$ & $\theta^A_{R4} = \frac{\pi}{2}$ & $\theta^D_{E4} = \frac{\pi}{4}$, $\phi^A_{R4} = \arctan\left(\frac{d_3}{d_1}\right)$ & \\ \hline

\end{tabular}
\end{table*}

\subsection{Adversarial Defense}

To improve the model’s robustness against adversarial attacks, we propose the adversarial training for multiple RIS-assisted MIMO Systems (ATMRM) algorithm. The goal is to reduce the system’s SER by mitigating the impact of adversarial perturbations. The adversarial training procedure for multiple RIS-assisted MIMO systems is outlined in Algorithm~\ref{alg:training}. The inputs to the algorithm include the model architecture $\mathcal{F}_{\theta}(.)$, noise variance $\sigma^2$, channel matrices $\mathbf{O}^i$, training dataset, training SNR ($\mathrm{SNR}_{\text{train}}$), PSR, and the adversarial perturbation vector $\mathbf{u}_{\mathrm{e}}$\footnote{\textcolor{black}{In practice, $\mathbf{u}_{\mathrm{e}}$ is not directly observed by the decoder but can be 
generated at the system design stage using AI models (e.g., large telecommunication 
models, GANs, diffusion) to simulate adversarial threats.}}.

\textcolor{black}{During each epoch, input batches along with $\mathbf{O}^i$ are loaded. For each batch, a fixed perturbation $\mathbf{u}_{\mathrm{e}}$ is applied at the decoder by modifying the received signal as $\mathbf{z}_i + \mathbf{O}^i \mathbf{u}_{\mathrm{e}}$. Note that $\mathbf{O}^i$ denotes the cascaded effective channel from the encoder through the distributed RISs to the decoder, which may include Doppler-induced variations under mobility. The model processes the perturbed input to produce the output, after which the BCE loss is evaluated. The model parameters are then updated via gradient descent to minimize the loss as follows:
\begin{equation}
\theta \leftarrow \theta - \alpha \nabla_{\theta} \mathcal{L}(\mathcal{F}_{\theta}(\mathbf{z}_i + \mathbf{O}^i \mathbf{u}_{\mathrm{e}}), y_i),
\end{equation}where $y_i$ is the correct label.} From a theoretical perspective, the adversarial training process can be modeled as a stochastic optimization problem:
\begin{equation}
\min_{\theta} \mathbb{E}_{\mathbf{O}^i, \mathbf{u}_{\mathrm{e}}} \left[\mathcal{L}\left(\mathcal{F}_{\theta}(\mathbf{z}_i + \mathbf{O}^i \mathbf{u}_{\mathrm{e}}), y_i\right)\right].
\label{eq:adv_train}
\end{equation}

Unlike conventional adversarial training algorithms, Algorithm 2 not only considers distributed RIS deployments with correlated fading but also can adapt with mobility. Upon convergence, the algorithm yields a model that demonstrates robust performance under both adversarial and dynamic wireless conditions. The practical effectiveness of the proposed defense is validated via numerical simulations across multiple RIS configurations and mobility scenarios, as discussed in the following section.

\label{NUMERICAL RESULTS}
\section{Numerical results}
\label{Numerical results}
\begin{figure*}[t]
	\centering
    \begin{minipage}[t]{0.31\textwidth}
	\includegraphics[width=1\textwidth]{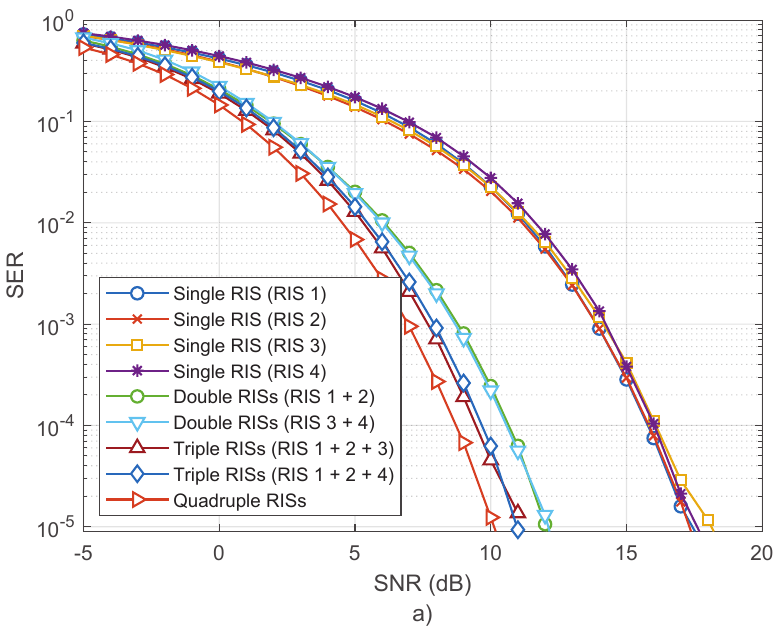}
    \end{minipage}
    \begin{minipage}[t]{0.31\textwidth}
	\includegraphics[width=1\textwidth]{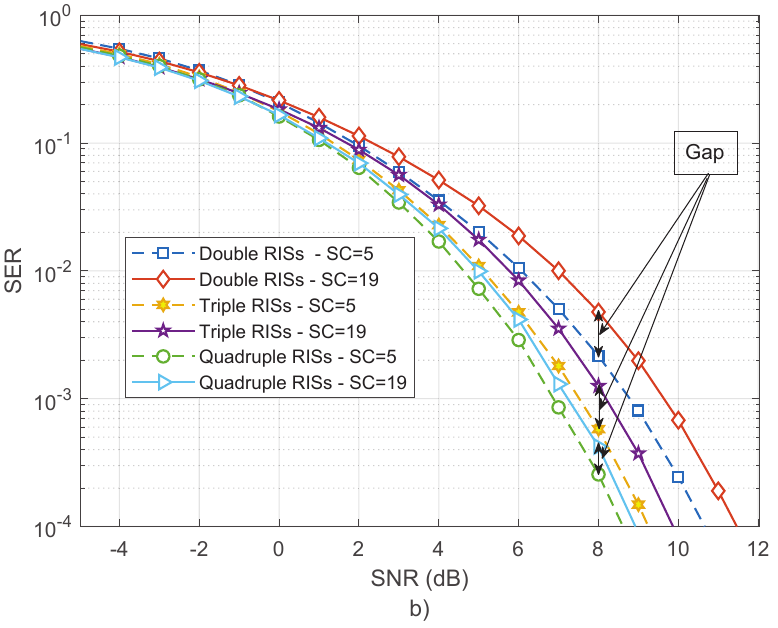}
    \end{minipage}
    	\centering
    \begin{minipage}[t]{0.31\textwidth}
	\includegraphics[width=1\textwidth]{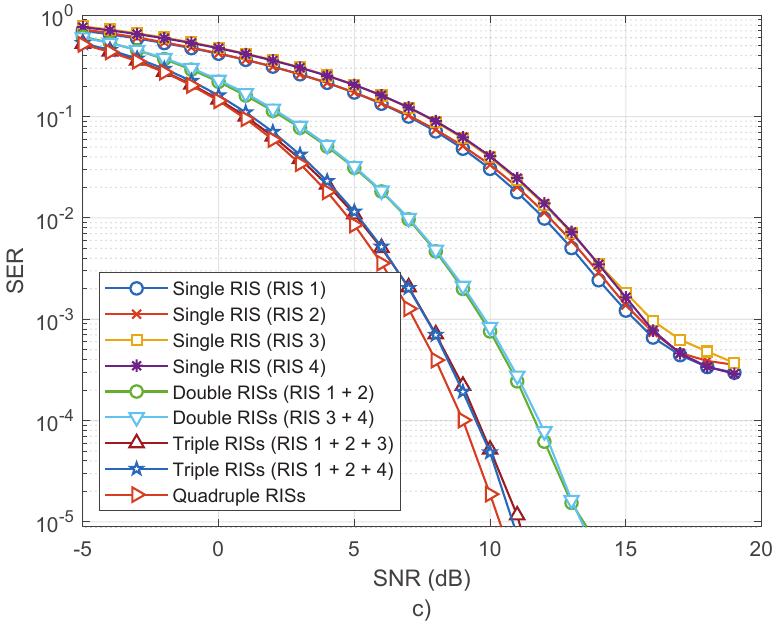}
    \end{minipage}
\caption{SER performance of RIS-assisted MIMO systems under double-scattering channels with varying numbers of scatterers: (a) $\mathbf{SC}_\mathbf{L}=5$, (b) $\mathbf{SC}_\mathbf{L}=5$ and $19$, and (c) $\mathbf{SC}_\mathbf{L}\to\infty$.}

\label{SERall}
\end{figure*}  

This section presents extensive numerical results\footnote{\textcolor{black}{ It is important to note that all of the Monte-Carlo simulations are investigated under the low-SNR regime.}} to evaluate the performance of the proposed adversarial attacks and defense. We also compare the robustness of multiple RIS-aided systems trained with and without adversarial training. 

\subsection{Simulation Configuration}
    
    
    
    
    
The deployment of the distributed RIS-assisted MIMO autoencoder system is structured in a three-dimensional (3D) Cartesian coordinate system to enhance signal propagation, spatial diversity, and practical deployment considerations. The encoder (transmitter) and decoder (receiver) are located at coordinates $(d_1, 0, d_3)$ and $(d_1, d_1, 0)$, respectively. Four distributed RISs are strategically placed as:  RIS~1 at $(d_2, 0, d_3)$ with azimuth angle $\frac{3\pi}{4}$; RIS~2 at $(0, 0, d_3)$ with azimuth angle $\frac{\pi}{4}$; RIS~3 at $(0, d_1, d_3)$ with azimuth angle $-\frac{\pi}{4}$; and RIS~4 at $(d_2, d_1, d_3)$ with azimuth angle $-\frac{3\pi}{4}$.  The path loss for all the communication links is modeled in accordance with the 3GPP Urban Micro (UMi) NLOS scenario~\cite{3GPP_TR36_814}, defined as $ \gamma_{\text{pl}} = 35.6 + 22\log_{10}(l)$, where $l$ is the link distance in meters and $\textcolor{black}{\alpha_m =\gamma_m = \gamma_{\text{pl}}}$.
The simulation parameters are configured as: transmit power: $P = 30$ [dBm]; noise power: $\sigma^2 = -90$ [dBm]; number of antennas at encoder and decoder: $K_e = K_d = 16$; and number of reflecting elements per RIS: $N_m = 32$.
An adversary equipped with $K_a = 16$ antennas is assumed to be located at $(d_1, 0, 0)$. The secure communication channels and the adversarial attack channels are comprehensively detailed in Table~\ref{locationRIS} and Table~\ref{locationadverasry}, \textcolor{black}{respectively}. \textcolor{black}{Additionally, throughout all simulations, $L_B$ is set to $20$.}

During the training phase, a dataset comprising 150,000 samples of data symbols and their corresponding channel realizations is generated. Among these, 135,000 samples are utilized for training, while the remaining 15,000 are reserved for testing and evaluation purposes. All 1D-CNN networks are optimized jointly using the Adam optimizer \cite{kingma2014adam}, with an initial learning rate set at 0.001, which is reduced by a factor of five every five epochs to facilitate stable convergence. Including 1D-BN layers facilitates rapid convergence, allowing us to complete the training in 20 epochs. During training, we vary the noise power to reflect different SNR levels of the transmitted signal, thereby evaluating the system’s robustness across diverse conditions.
All simulations are executed on a workstation equipped with an Intel Core i7-12700K CPU running at 3.6~GHz and an NVIDIA GeForce RTX 3060 GPU with 16~GB of memory. The Monte Carlo simulations are conducted in MATLAB R2024b, while the neural network models are implemented using Python 3.9.3 and the PyTorch deep learning framework.

The encoder takes an input of dimensions $M \times L_B$ and comprises three consecutive convolutional layers with 1D-BN and ReLU activation. The first two convolutional layers employ 256 filters of kernel size 1, while the final layer increases the filters to $2K_e$, producing an output of $2K_e \times L_B$. The decoder processes an input of size $2K_eL_B + 2K_eK_dL_B$ and consists of three convolutional layers. The first two layers utilize 512 filters with BN and ReLU activation, while the final layer applies a softmax activation with $M$ filters, generating an output of $M \times L_B$. The RIS model for each intelligent surface starts with an input of $2N_m \times L_B$ and consists of three convolutional layers. The first two layers have 512 filters with BN and ReLU activation, whereas the final layer uses $N_m$ filters to produce an output of $N_m\times L_B$. These configurations ensure efficient feature extraction and transformation across the encoder, decoder, and RIS models. \textcolor{black}{ The summary of parameter settings is presented in Table~\ref{tab:parra}}.


\subsection{Secured SER Evaluation}
\begin{figure*}[t]
	\centering
    \begin{minipage}[t]{0.31\textwidth}
	\includegraphics[width=1\textwidth]{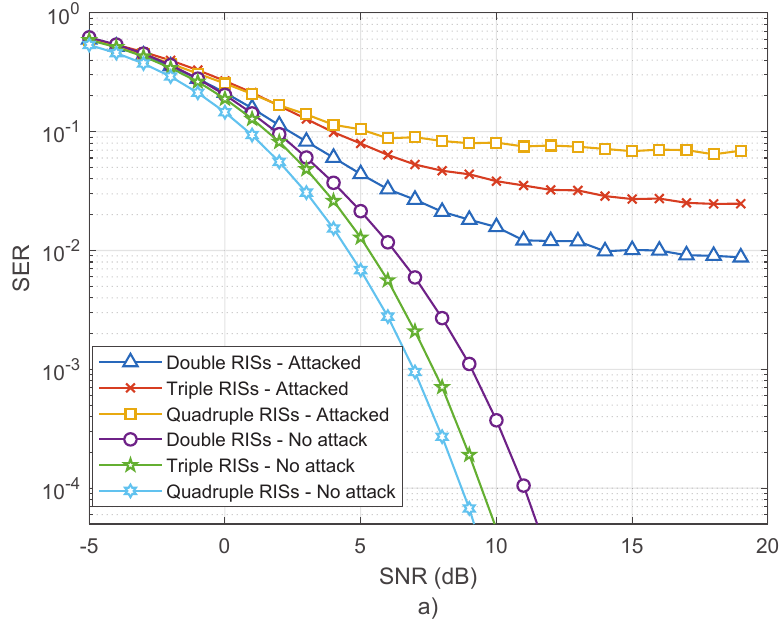}
    \end{minipage}
    \begin{minipage}[t]{0.31\textwidth}
	\includegraphics[width=1\textwidth]{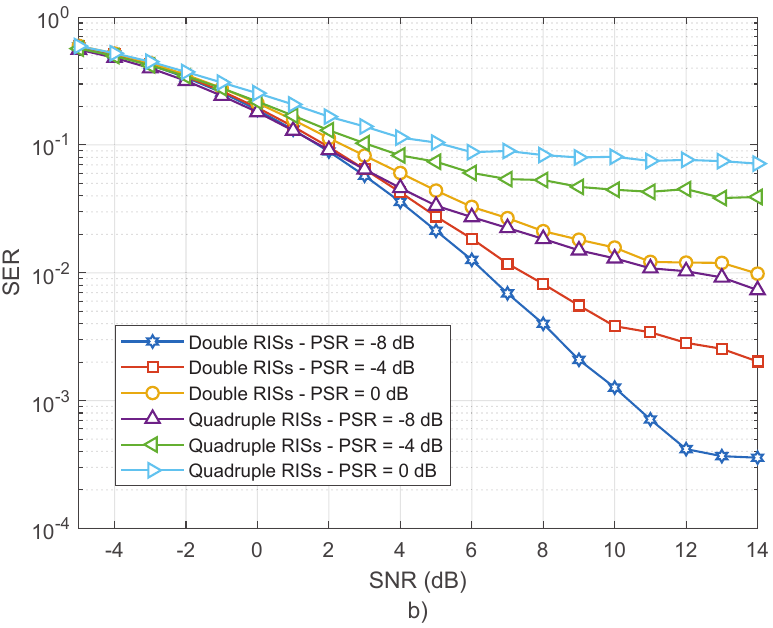}
    \end{minipage}
    \begin{minipage}[t]{0.31\textwidth}
	\includegraphics[width=1\textwidth]{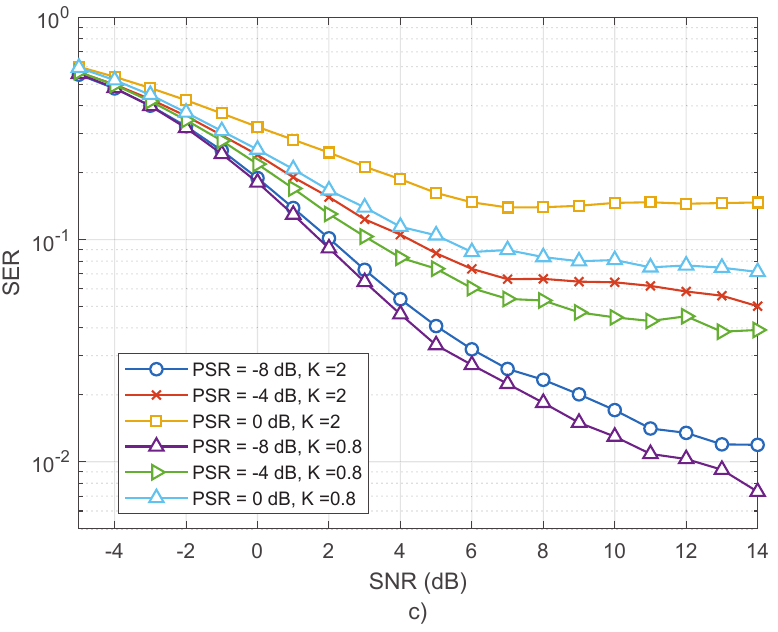}
    \end{minipage}

\caption{The SER performance of multiple RIS-assisted MIMO systems under the MRMAEF adversarial attack: (a) The SER of different RIS configurations at a fixed Rician factor of 0.8 and PSR of 0~dB, (b) The SER comparison of double and quadruple RISs at Rician factor 0.8 with varying PSR levels, and (c) The SER of the quadruple RIS configuration under varying Rician factors and PSR values.}
\label{MRMAEF}
\end{figure*}
\begin{figure*}[t]
	\centering
    \begin{minipage}[t]{0.31\textwidth}
	\includegraphics[width=1\textwidth]{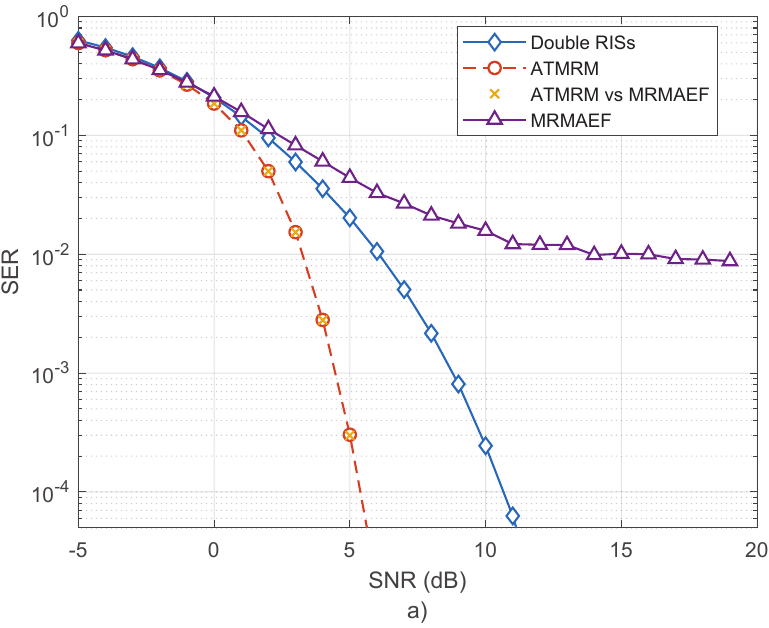}
    \end{minipage}
    \begin{minipage}[t]{0.31\textwidth}
	\includegraphics[width=1\textwidth]{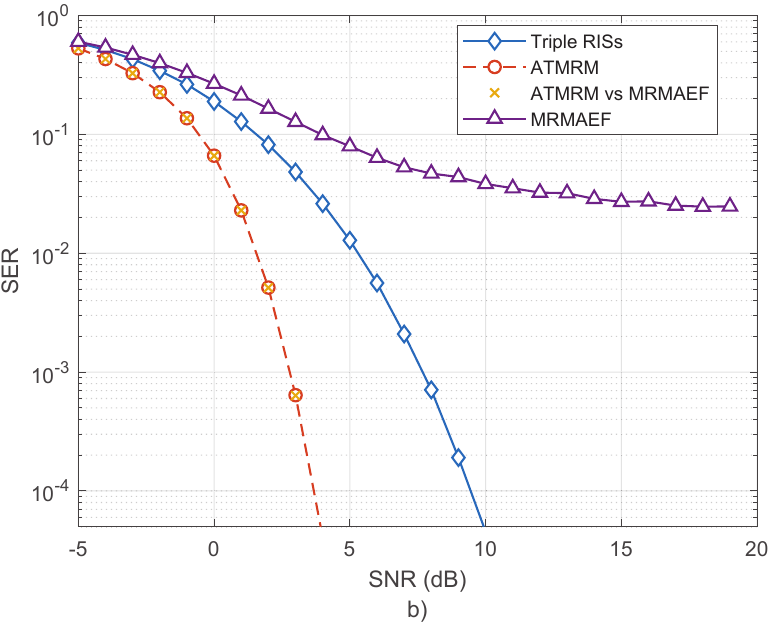}
    \end{minipage}
    \begin{minipage}[t]{0.31\textwidth}
	\includegraphics[width=1\textwidth]{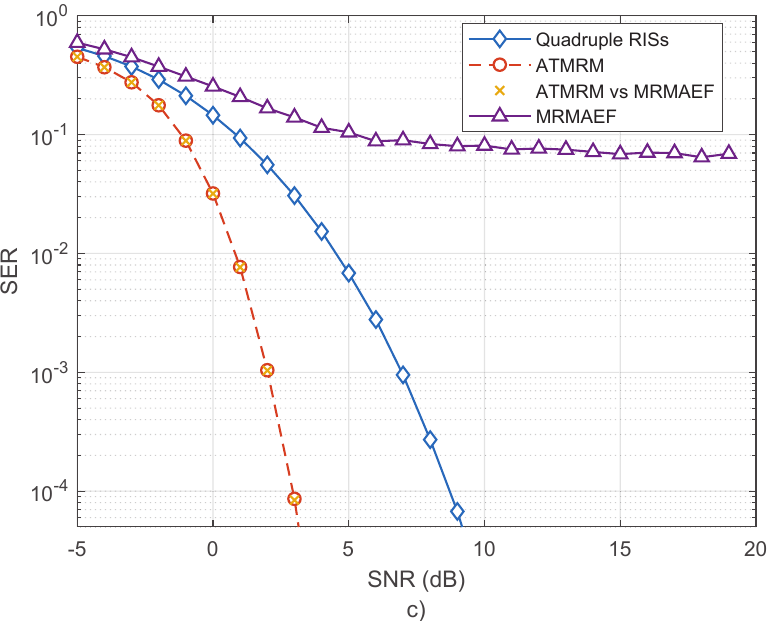}
    \end{minipage}
\caption{The SER performance of RIS-assisted MIMO systems under adversarial attack (MRMAEF) and defense (ATMRM) over double-scattering channels with 5 scatterers: (a) Double RISs, (b) Triple RISs, and (c) Quadruple RISs.}
\label{defense1}
\end{figure*}
\label{Secured_SER_evaluation}
\textcolor{black}{The practical deployment of distributed RISs presents several challenges, including the overhead of synchronization between RIS units, the complexity of coordination, and the need for efficient control signaling. This study examines static transceiver positions, providing a useful baseline for analyzing and optimizing SER performance in both adversarial and non-adversarial conditions\footnote{\textcolor{black}{It is important to note that imperfect CSI acquisition can degrade the SER performance, especially in traditional systems with multiple distributed RISs, where obtaining accurate CSI becomes increasingly challenging due to limited operational power and higher coordination complexity. In contrast, the proposed autoencoder-based framework enables joint data encoding and RIS phase shift optimization without relying on explicit CSI \cite{10136735}, thereby reducing the feedback overhead typically required for downlink channel estimation.}}.} The SER is evaluated using an end-to-end approach with $d_1 = 200$~[m], $d_2 = 400$~[m], and $d_3 = 3$~[m]. The channel parameters are defined as $\mathrm{SC}_\mathbf{L} = \{5, 11, 19\}$, where $\mathbf{L}$ belongs to $\{\mathbf{D}_m, \mathbf{H}_m\}$ for $m = 1, 2, 3, 4$. Additionally, the Rician factor (K) is set to 4 for all links. 

Fig.~\ref{SERall}(a) shows that using multiple distributed RISs results in a significant performance improvement compared to the case of using only a single RIS. In particular, the configuration with four RISs achieves the lowest SER across all SNR values, highlighting the advantage of utilizing multiple distributed RISs to enhance signal propagation and improve transmission reliability. Specifically, at the SNR of 10~[dB], the SER of four RISs setup is approximately $10^{-5}$, which is significantly lower than three RISs setups using RIS~1,~2,~3 and RIS~1,~2,~4, at around $6 \times 10^{-5}$. At this SNR, for double-RIS configurations (RIS~1,~2 and RIS~3,~4), achieves the SER approximately $4 \times 10^{-4}$. For the single RIS, the system's SER is significantly lower than that of multiple RIS configurations. It is worth noting that the slight variations in SER among different single, double, and triple RIS combinations are attributed to the stochastic nature of the simulated dataset. Despite this, both the triple- and quadruple-RIS configurations show rapid SER reduction and achieve sufficiently low SER when the SNR exceeds 10~[dB]. 
In a double-scattering fading channel, the number of scatterers significantly impacts the system's performance \cite{7790363}. In practical scenarios, the number of scatter points varies \textcolor{black}{depending on} several factors such as building materials, obstruction geometry, and the surrounding propagation environment. For more details, we provide Fig.~\ref{SERall}(b) to show the gap of the SER between $\mathrm{SC}_\mathbf{L} =5$ and $\mathrm{SC}_\mathbf{L} =19$ scenarios for different RIS configurations. As shown, the SER gap for the system with triple RISs is larger than that of the quadruple-RIS setup. Notably, the gap observed in the quadruple-RIS case is approximately half of that in the double-RIS system. Consequently, in environments with unpredictable or harsh scattering characteristics, using additional RISs can significantly enhance communication reliability. As the $\mathbf{SC}_\mathbf{L} \to \infty$, the channel coefficients are statistically independent, \textcolor{black}{turning into} the uncorrelated Rayleigh‐fading model \cite{7790363}. To evaluate the model’s behavior in the infinity scattering regime, we provide Fig.~\ref{SERall}(c), which shows the SER performance as $\mathbf{SC}_\mathbf{L} \to \infty$. As $\mathbf{SC}_\mathbf{L}$ grows without bound, the single-RIS configuration reduces SER-minimization capability compared to $\mathbf{SC}_\mathbf{L} = 5$. The double, triple, and quadruple-RIS setups follow the same trend; however, for very large $\mathbf{SC}_\mathbf{L}$, the triple- and quadruple-RIS configurations achieve nearly the same SER reduction as observed at $\mathbf{SC}_\mathbf{L} = 5$. \textcolor{black}{Additionally, Fig.~\ref{SERall}(a) and (c) show that in the low-SNR regime, fewer scatterers concentrate the received energy along keyhole-like paths, which helps the decoder overcome noise and achieve lower SER. This observation is consistent with the theoretical analyses of 
double-scattering channels\cite{10136735,9531522}.}


\subsection{Adversarially  Attacked SER Evaluation}
\label{attack_evaluation}
The SER is evaluated with  the system's locations as the same as Section~\ref{Secured_SER_evaluation}. We fix $\mathrm{SC}_{\mathbf{L}} = 5$, representing a sparse scattering environment. Regarding RIS deployment, we consider three configurations: a double-RIS setup utilizing RIS~1 and RIS~2; a triple-RIS setup including RIS~1, RIS~2, and RIS~3; and a quadruple-RIS setup employing all four RISs. These configurations allow us to analyze the system's robustness against adversarial perturbations under a varying number of distributed RISs.

Fig.~\ref{MRMAEF}(a) illustrates Algorithm~\ref{alg:2} performance regarding maximizing the SER of the proposed system, considering K of the attack channels of $0.8$ and PSR of $ 0$~[dB]. It increases the SER of the quadruple-RIS-aid system up to approximately $8 \times 10^{-2}$, while triple- and double-RIS configurations converge to around $5 \times 10^{-2}$ and $10^{-2}$, respectively. It demonstrates that when the number of RISs is increased, despite improving baseline SER, the system becomes more vulnerable to adversarial perturbations. Fig.~\ref{MRMAEF}(b) shows that reducing the PSR decreases the effectiveness of the attack. For the two-RIS configuration, reducing PSR from $0$~[dB] to $-8$ [dB] lowers the SER from about $10^{-2}$ to $3 \times 10^{-4}$. For the quadruple‐RIS–assisted system, the SER under attack decreases from nearly $9\times10^{-2}$ at PSR $= 0$ [dB] to about $10^{-2}$ at PSR $= -8$ [dB], when SNR is $14$~[dB]. Finally, Fig. 4(c) illustrates the effect of K in the attacker’s channels. When fixing PSR of $0$~[dB] and SNR of $12$~[dB], increasing K from 0.8 to 2 raises the SER from approximately $8 \times 10^{-2}$ to $2 \times 10^{-1}$. Even at the PSR of $4$~[dB] the SER grows significantly from $5 \times 10^{-2}$ to around $6.5 \times 10^{-2}$ for K=$0.8$ and K=$2$, respectively.  This confirms that stronger dominant paths in the attack channel make the system less robust under adversarial conditions.
\begin{table}[t]
\caption{Summary of Paramenters}
\label{tab:parra}
\centering
\textcolor{black}{
\resizebox{\columnwidth}{!}{%
\begin{tabular}{|c|c|c|cc|c|}
\hline
\multicolumn{1}{|c|}{\multirow{2}{*}{Layer name}} & \multicolumn{1}{c|}{\multirow{2}{*}{\begin{tabular}[c]{@{}c@{}}Activation\\ function\end{tabular}}} & \multicolumn{2}{c|}{Parameters}                           & \multirow{2}{*}{\begin{tabular}[c]{@{}c@{}}Output \\ dimensions\end{tabular}} \\ \cline{3-4}
\multicolumn{1}{|c|}{}                            & \multicolumn{1}{c|}{}                                                                               & \multicolumn{1}{c|}{Kernel} & \multicolumn{1}{c|}{Filter} &                                                                               \\ \hline
\multicolumn{5}{|c|}{Encoder}                                                                                                                                                                                                                                                                       \\ \hline
\multicolumn{1}{|c|}{Input}                       & \multicolumn{1}{c|}{None}                                                                           & \multicolumn{2}{c|}{None}                                 & $M \times L_B$                                                                \\ \hline
\multicolumn{1}{|c|}{Conv1D + BN}                 & \multicolumn{1}{c|}{ReLU}                                                                           & \multicolumn{1}{c|}{1}      & \multicolumn{1}{c|}{256}    & $256 \times L_B$                                                              \\ \hline
\multicolumn{1}{|c|}{Conv1D + BN}                 & \multicolumn{1}{c|}{ReLU}                                                                           & \multicolumn{1}{c|}{1}      & \multicolumn{1}{c|}{256}    & $256 \times L_B$                                                              \\ \hline
\multicolumn{1}{|c|}{Conv1D + BN}                 & \multicolumn{1}{c|}{None}                                                                           & \multicolumn{1}{c|}{1}      & \multicolumn{1}{c|}{$K_e$}  & $2K_e \times L_B$                                                              \\ \hline
\multicolumn{5}{|c|}{RIS m (m = 1, 2, 3, 4)}                                                                                                                                                                                                                                                        \\ \hline
\multicolumn{1}{|c|}{Input}                       & \multicolumn{1}{c|}{None}                                                                           & \multicolumn{2}{c|}{None}                                 & $2N_m \times L_B$                                                             \\ \hline
\multicolumn{1}{|c|}{Conv1D +BN}                  & \multicolumn{1}{c|}{ReLU}                                                                           & \multicolumn{1}{c|}{1}      & \multicolumn{1}{c|}{512}    & $512 \times L_B$                                                              \\ \hline
\multicolumn{1}{|c|}{Conv1D + BN}                 & \multicolumn{1}{c|}{ReLU}                                                                           & \multicolumn{1}{c|}{1}      & \multicolumn{1}{c|}{512}    & $512 \times L_B$                                                              \\ \hline
\multicolumn{1}{|c|}{Conv1D}                      & \multicolumn{1}{c|}{None}                                                                           & \multicolumn{1}{c|}{1}      & \multicolumn{1}{c|}{$N_m$}      & $N_m \times L_B$                                                              \\ \hline
\multicolumn{5}{|c|}{Decoder}                                                                                                                                                                                                                                                                       \\ \hline
\multicolumn{1}{|c|}{Input}                       & \multicolumn{1}{c|}{None}                                                                           & \multicolumn{2}{c|}{None}                                 & $2(K_e +K_eK_d) \times L_B$                                                   \\ \hline
\multicolumn{1}{|c|}{Conv1D + BN}                 & \multicolumn{1}{c|}{ReLU}                                                                           & \multicolumn{1}{c|}{1}      & \multicolumn{1}{c|}{512}    & $512 \times L_B$                                                              \\ \hline
\multicolumn{1}{|c|}{Conv1D + BN}                 & \multicolumn{1}{c|}{ReLU}                                                                               & \multicolumn{1}{c|}{1}      & \multicolumn{1}{c|}{512}    & $512 \times L_B$                                                              \\ \hline
\multicolumn{1}{|c|}{Conv1D + BN}                 & \multicolumn{1}{c|}{Softmax}                                                                        & \multicolumn{1}{c|}{1}      & \multicolumn{1}{c|}{$M$}    & $M \times L_B$       
\\ \hline
\end{tabular}%
}}
\end{table}

\textcolor{black}{We provide Table~\ref{tab:fgm_pgd_generation} to compare the UAP generation times of MRMAEF and MRMAEP under the quadruple RISs configuration. It reveals that MRMAEP generation is significantly more time-consuming than MRMAEF, with a total runtime almost 50 times higher (25750.26~[s] versus 524.55~[s]). Additionally, the average UAP generation time per SNR for PGD exceeds 1030 [s], compared to approximately 21 [s] for MRMAEF. Although MRMAEP provides stronger adversarial perturbations, this result highlights the substantial computational overhead involved, suggesting that MRMAEF may be more suitable for time-sensitive or large-scale evaluations, while MRMAEP is better suited for thorough robustness testing.}

\subsection{Adversarial Defense Symbol Error Rate Evaluation}
\begin{figure*}[t]
	\centering
    \begin{minipage}[t]{0.31\textwidth}
	\includegraphics[width=1\textwidth]{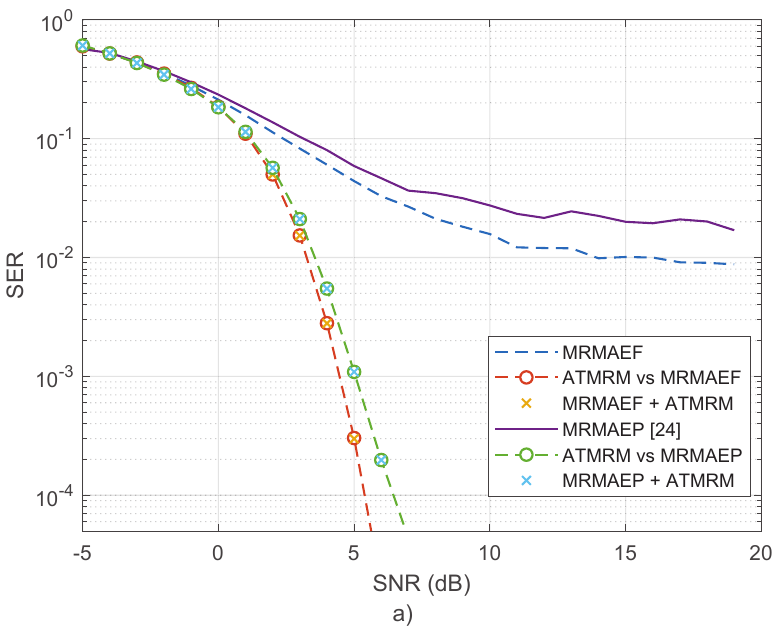}
    \end{minipage}
    \begin{minipage}[t]{0.31\textwidth}
	\includegraphics[width=1\textwidth]{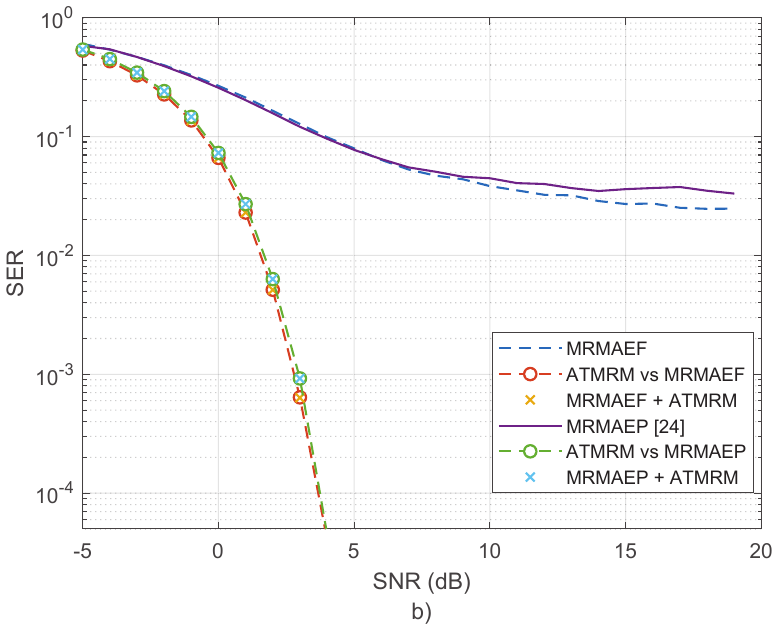}
    \end{minipage}
    \begin{minipage}[t]{0.31\textwidth}
	\includegraphics[width=1\textwidth]{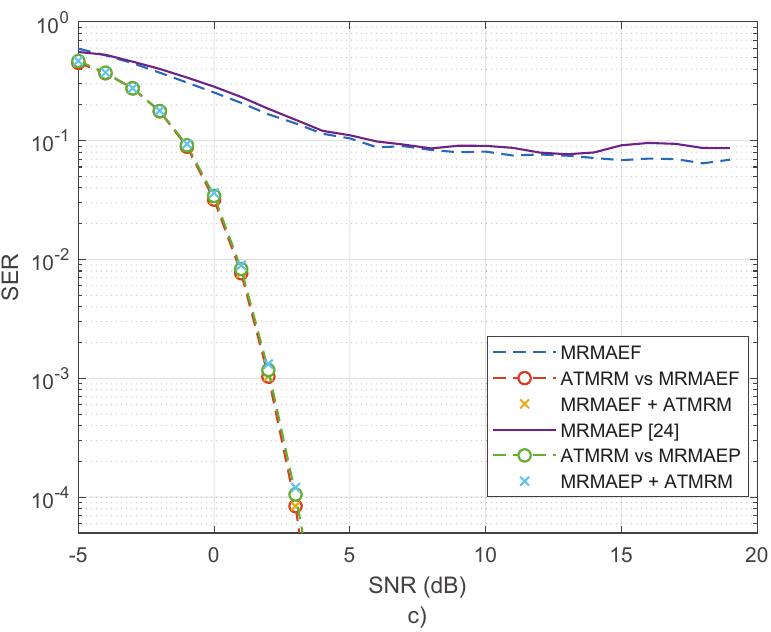}
    \end{minipage}
\caption{The SER performance of RIS-assisted MIMO systems under adversarial attacks (MRMAEF and MRMAEP) and defense (ATMRM) over double-scattering channels with 5 scatterers: (a) Double RISs, (b) Triple RISs, and (c) Quadruple RISs.}
\label{defense2}
\end{figure*}
\begin{figure*}[t]
	\centering
    \begin{minipage}[t]{0.24\textwidth}
	\includegraphics[width=1\textwidth]{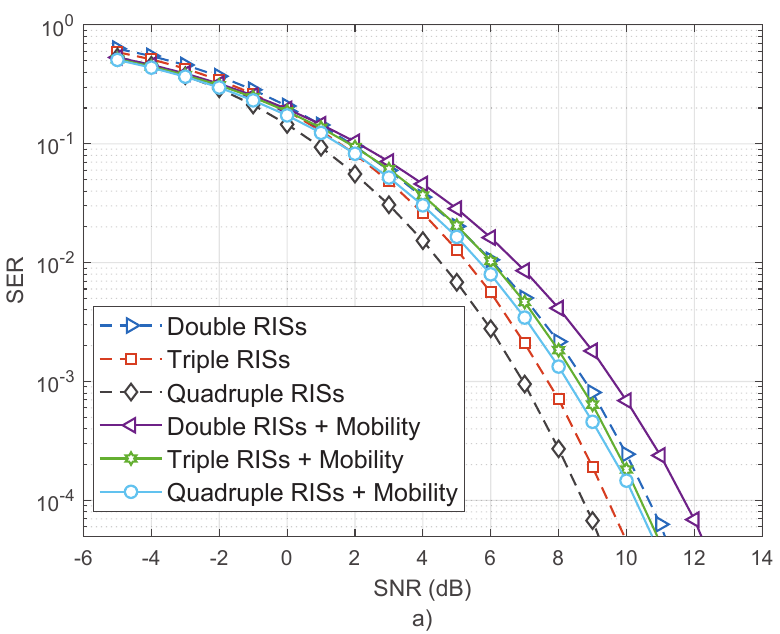}
    \end{minipage}
    \begin{minipage}[t]{0.24\textwidth}
	\includegraphics[width=1\textwidth]{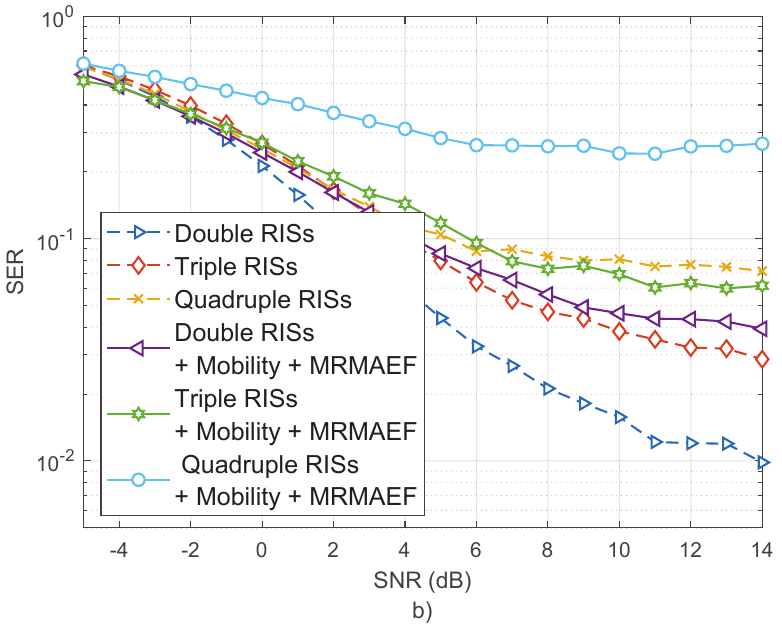}
    \end{minipage}
    \begin{minipage}[t]{0.24\textwidth}
	\includegraphics[width=1\textwidth]{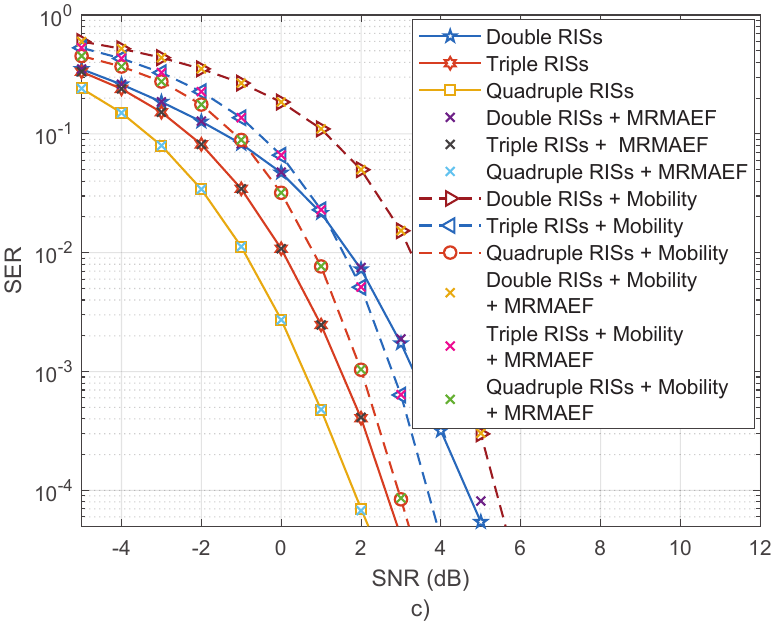}
    \end{minipage}
    \begin{minipage}[t]{0.24\textwidth}
	\includegraphics[width=1.07\textwidth]{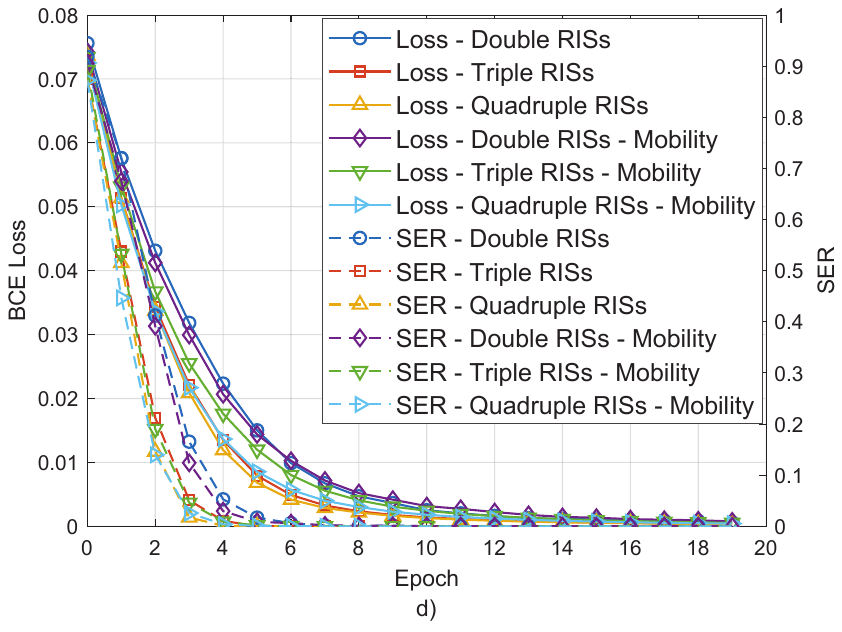}
    \end{minipage}    
\caption{\textcolor{black}{Evaluation under mobile and static scenarios with varying numbers of RISs: (a) The SER performance of the proposed model, (b) The SER performance of victim models under adversarial attacks generated by MRMAEF, (c) The SER comparison between ATMRM and non-ATMRM, and (d) Training loss and SER convergence of ATMRM.}}
\label{losss}
\end{figure*}
\label{Adversarial Defense}
The SER under adversarial defense strategies using an end-to-end learning framework is now evaluated with the number of scatterers as $\mathrm{SC}_\mathbf{L} = \{5, 11, 19\}$. The RIS deployment is the same section~\ref{attack_evaluation}. These scenarios enable us to evaluate the effectiveness of defense mechanisms against adversarial attacks with varying numbers of distributed RISs.

Fig.~\ref{defense1} compares the SER of the proposed system under the MRMAEF attack, with and without the ATMRM defense, for $\mathbf{SC}_\mathbf{L} = 5$.  Fig.~\ref{defense1}(a) demonstrates that with Algorithm~\ref{alg:training}, the system SER achieves $5 \times 10^{-5}$ at the SNR of $5.5$~[dB] versus 11 [dB] without defense, even under the effect of the MRMAEF attack. Additionally, Fig.~\ref{defense1}(b) shows that, with ATMRM, the triple‐RIS configuration under attack achieves the SER of approximately $10^{-4}$ at SNR of $4$~[dB] compared to $10$~[dB] without defense. Similarly, for the quadruple‐RIS setup in Fig.~\ref{defense1}(c), the system achieves an SER of $10^{-4}$ with the SNR of about $3$~[dB] versus $8.5$~[dB] without ATMRM. These results demonstrate that Algorithm~\ref{alg:training} is not only effective against MRMAEF but also significantly improves SER performance as the number of RISs increases. \textcolor{black}{This improvement arises from the ATMRM algorithm, which trains the autoencoder on adversarial perturbations and thereby enhances its robustness. As a result, since the model is trained with adversarial perturbations rather than only clean data, the decoder learns a more resilient mapping, which enables the SER performance under MRMAEF with ATMRM to surpass that of the model trained only on clean data.} 
Fig.~\ref{defense2}(a) shows the SER under MRMAEF and MRMAEP attacks, both with and without the proposed ATMRM defense mechanism, in the double-RIS configuration. Without defense, both attacks significantly degrade reliability and raise the SER. In this setup, ATMRM is more effective against MRMAEF than MRMAEP. Specifically, at an SNR of 5~[dB], the SER under MRMAEP is about $10^{-3}$, while under MRMAEF it is approximately $5 \times 10^{-4}$. Fig.~\ref{defense2}(b) and Fig.~\ref{defense2}(c) extend the analysis to the triple-RIS and quadruple setups. In both setups, ATMRM continues to provide strong protection and reduces the SER compared to the undefended system. The defense gap between MRMAEF and MRMAEP attacks becomes identical, i.e, the SER of ATMRM under these attacks is nearly the same. These results suggest that, in addition to enhancing spatial diversity, increasing the number of RISs improves the system’s ability to learn and defend against adversarial attacks when ATMRM is employed.
\begin{table}[t]
\centering
\caption{\textcolor{black}{UAP Generation Time: MRMAEF vs. MRMAEP under Quadruple RISs Configuration}}
\label{tab:fgm_pgd_generation}
\textcolor{black}{{\begin{tabular}{|c|c|c|}
\hline
\textbf{Metric} & \textbf{MRMAEF} & \textbf{MRMAEP} \\ 
\hline
Total Generation Time & 524.55~s & 25750.26~s \\
Average Time per SNR & 20.98~s & 1030.01~s \\
Fastest SNR & 9.52~s & 307.90~s \\
Slowest SNR & 23.54~s & 1323.50~s\\
\hline
\end{tabular}}}
\end{table}
\begin{table}[t]
\centering
\caption{\textcolor{black}{Comparison of Training Time With and Without ATMRM under Quadruple RISs Configuration}}
\label{tab:training_time_comparison}
\textcolor{black}{
{
\begin{tabular}{|c|c|c|}
\hline
\textbf{Metric} & \textbf{With ATMRM} & \textbf{Without ATMRM} \\ 
\hline
Total Training Time & 93.61~s  & 89.27~s  \\
Average Time per Epoch & 4.68~s & 4.46~s \\
Fastest Epoch & 4.42~s & 4.38~s \\
Slowest Epoch & 8.16~s & 4.81~s \\
Average Time per Step & 0.103~s & 0.096~s \\
\hline
\end{tabular}}}
\end{table}

\textcolor{black}{Table~\ref {tab:training_time_comparison} presents a detailed comparison of the training time with and without adversarial training under the quadruple RISs configuration. It illustrates that adversarial training introduces only a minor computational overhead of approximately 
$4.9\%$ in terms of total training time. The average epoch duration increases slightly from $4.46$~[s] to $4.68$~[s], indicating that the proposed adversarial training framework achieves enhanced model robustness with a negligible impact on training efficiency. The variance between the fastest and slowest epochs also remains within an acceptable range, demonstrating stable convergence behavior.}

\subsection{\textcolor{black}{Adversarial Attack and Defense Under The Doppler Effect}}
\textcolor{black}{In this subsection, we investigate the SER performance under mobility, where the destination (i.e., decoder) is in motion. The channel links from the $m^{th}$ RISs to the mobile decoder are modeled as
\begin{equation}
    \mathbf{L}_{\mathrm{d}} = \pmb{\zeta}\sqrt{\epsilon_m}\left( \sqrt{\frac{\alpha_m}{\alpha_m+1}}\mathbf{\overline{\mathbf{H}}_m} +  \sqrt{\frac{1}{\alpha_m+1}}\mathbf{\hat{\mathbf{H}}_m}\right) + \sqrt{1-(\pmb{\zeta})^2}\pmb{\xi}
\end{equation}
where $\pmb{\zeta} = J_0^2(2\pi f_D T_b)$ is the temporal correlation coefficient determined by the squared zeroth-order Bessel function of the first kind. The Doppler frequency is given by $f_D = \frac{Vf_0}{c}$, where $V$ denotes the velocity of the moving terminal (decoder), $f_0$ is the carrier frequency, $c$ is the speed of light, and $T_b$ is the symbol sampling interval. $\pmb{\xi} \sim \mathcal{CN}(0, \sigma^2\mathbf{B}_{K_d})$ denotes an independent random fading component that accounts for mobility-induced randomness.} \textcolor{black}{For the simulations, we set the decoder's velocity to $20$ [m/s] and the carrier frequency is $2.6$ [GHz], which leads to the Doppler frequency of approximately $173.3$ [Hz].}

As shown in Fig.~\ref{losss}(a), decoder mobility leads to performance degradation, where a higher SNR is required to achieve the same SER compared to the static scenario, due to channel variations induced by mobility. Fig.~\ref{losss}(b) presents the SER performance of the MRMAEF attack across different numbers of RISs under both static and mobile scenarios. In the presence of mobility, the model demonstrates increased vulnerability compared to the static case. This degradation is primarily attributed to the Doppler effect and the resulting time-varying nature of double-scattering channels. The adversary's attack capability is enhanced under mobility, as Doppler-induced channel fluctuations significantly amplify the effect of adversarial perturbations. To evaluate the performance of ATMRM under decoder mobility, Fig.~\ref{losss}(c) is presented. In the absence of ATMRM, achieving an SER of $10^{-4}$ in the double, triple, and quadruple RIS configurations requires a significantly higher SNR compared to the static scenario. In contrast, the ATMRM-enhanced model exhibits improved robustness under mobility and even achieves a lower SER than the model operating in the static case. This phenomenon can be attributed to the model’s ability to learn more effective decoding strategies in the presence of signal mismatch induced by mobility-related channel variations. Fig.~\ref{losss}(d) illustrates the training behavior of the proposed ATMRM method in terms of SER and BCE loss across different RIS configurations in both mobile and static scenarios. In all six cases, the SER and BCE loss exhibit a rapid decline during the initial training epochs and converge toward zero, indicating the effectiveness, training stability, and robustness of the proposed method, particularly with an increasing number of RISs.
\section{Conclusion}
\label{conclusion}
This paper has explored the adversarial robustness of distributed multiple RIS-assisted MIMO autoencoder systems under finite scattering environments. We first established a closed-form channel model for the aggregated link with multiple RISs and investigated the system performance in terms of SER. Our analysis revealed that while the integration of additional RISs significantly improves the baseline communication performance, it also introduces increased vulnerability to adversarial attacks. To address this, we proposed an adversarial training-based defense method tailored to the multiple distributed RIS scenario. Simulation results confirmed the effectiveness of the proposed defense approach, showing that it not only mitigates the impact of adversarial attacks, even in challenging white-box settings, but also enhances the SER in attack-free conditions. These findings highlight the scalability and robust training in building resilient intelligent communication systems. \textcolor{black}{Future research can further enhance the practicality and security of distributed RIS-assisted systems by jointly optimizing communication and control aspects. Potential directions include energy-efficient power control strategies, intelligent RIS placement and phase-shift design to improve coverage and robustness, and an in-depth study of black-box and transfer-based adversarial attacks. Such extensions will help strengthen the resilience of next-generation RIS-assisted communication networks against both known and unseen threats.}





\bibliographystyle{IEEEtran}
\bibliography{IEEE}

\begin{thebibliography}{10}
\providecommand{\url}[1]{#1}
\csname url@samestyle\endcsname
\providecommand{\newblock}{\relax}
\providecommand{\bibinfo}[2]{#2}
\providecommand{\BIBentrySTDinterwordspacing}{\spaceskip=0pt\relax}
\providecommand{\BIBentryALTinterwordstretchfactor}{4}
\providecommand{\BIBentryALTinterwordspacing}{\spaceskip=\fontdimen2\font plus
\BIBentryALTinterwordstretchfactor\fontdimen3\font minus
  \fontdimen4\font\relax}
\providecommand{\BIBforeignlanguage}[2]{{%
\expandafter\ifx\csname l@#1\endcsname\relax
\typeout{** WARNING: IEEEtran.bst: No hyphenation pattern has been}%
\typeout{** loaded for the language `#1'. Using the pattern for}%
\typeout{** the default language instead.}%
\else
\language=\csname l@#1\endcsname
\fi
#2}}
\providecommand{\BIBdecl}{\relax}
\BIBdecl

\bibitem{shi2024ris}
E.~Shi, J.~Zhang, H.~Du, B.~Ai, C.~Yuen, D.~Niyato, K.~B. Letaief, and X.~Shen,
  ``Ris-aided cell-free massive mimo systems for 6g: Fundamentals, system
  design, and applications,'' \emph{Proc. IEEE}, vol. 112, no.~4, pp. 331--364,
  2024.

\bibitem{9475160}
C.~Pan, H.~Ren, K.~Wang, J.~F. Kolb, M.~Elkashlan, M.~Chen, M.~Di~Renzo,
  Y.~Hao, J.~Wang, A.~L. Swindlehurst, X.~You, and L.~Hanzo, ``Reconfigurable
  intelligent surfaces for 6g systems: Principles, applications, and research
  directions,'' \emph{IEEE Commun. Mag.}, vol.~59, no.~6, pp. 14--20, 2021.

\bibitem{tang2020mimo}
W.~Tang, J.~Y. Dai, M.~Z. Chen, K.-K. Wong, X.~Li, X.~Zhao, S.~Jin, Q.~Cheng,
  and T.~J. Cui, ``Mimo transmission through reconfigurable intelligent
  surface: System design, analysis, and implementation,'' \emph{IEEE J. Sel.
  Areas Commun.}, vol.~38, no.~11, pp. 2683--2699, 2020.

\bibitem{khoshafa2024ris}
M.~H. Khoshafa, O.~Maraqa, J.~M. Moualeu, S.~Aboagye, T.~M. Ngatched, M.~H.
  Ahmed, Y.~Gadallah, and M.~Di~Renzo, ``Ris-assisted physical layer security
  in emerging rf and optical wireless communication systems: A comprehensive
  survey,'' \emph{IEEE Commun. Surveys Tuts.}, 2024.

\bibitem{9360873}
H.~Ye, G.~Y. Li, and B.-H. Juang, ``Deep learning based end-to-end wireless
  communication systems without pilots,'' \emph{IEEE Trans. Cogn. Commun.
  Netw.}, vol.~7, no.~3, pp. 702--714, 2021.

\bibitem{zhao2024generative}
C.~Zhao, H.~Du, D.~Niyato, J.~Kang, Z.~Xiong, D.~I. Kim, X.~Shen, and K.~B.
  Letaief, ``Generative ai for secure physical layer communications: A
  survey,'' \emph{IEEE Trans. Cogn. Commun. Netw.}, 2024.

\bibitem{8651357}
M.~Sadeghi and E.~G. Larsson, ``Physical adversarial attacks against end-to-end
  autoencoder communication systems,'' \emph{IEEE Commun. Lett.}, vol.~23,
  no.~5, pp. 847--850, 2019.

\bibitem{8449065}
------, ``Adversarial attacks on deep-learning based radio signal
  classification,'' \emph{IEEE Wireless Commun. Lett.}, vol.~8, no.~1, pp.
  213--216, 2019.

\bibitem{9838452}
Y.~Dong, H.~Wang, and Y.-D. Yao, ``A robust adversarial network-based
  end-to-end communications system with strong generalization ability against
  adversarial attacks,'' in \emph{Proc. IEEE Int. Conf. Commun.}, 2022, pp.
  4086--4091.

\bibitem{zheng2023designing}
K.~Zheng and X.~Ma, ``Designing learning-based adversarial attacks to (mimo-)
  ofdm systems with adaptive modulation,'' \emph{IEEE Trans. Wireless Commun.},
  vol.~22, no.~9, pp. 6241--6251, 2023.

\bibitem{nan2023physical}
G.~Nan, Z.~Li, J.~Zhai, Q.~Cui, G.~Chen, X.~Du, X.~Zhang, X.~Tao, Z.~Han, and
  T.~Q. Quek, ``Physical-layer adversarial robustness for deep learning-based
  semantic communications,'' \emph{IEEE J. Sel. Areas Commun.}, vol.~41, no.~8,
  pp. 2592--2608, 2023.

\bibitem{10138665}
P.~F. de~Araujo-Filho, G.~Kaddoum, M.~Chiheb Ben~Nasr, H.~F. Arcoverde, and
  D.~R. Campelo, ``Defending wireless receivers against adversarial attacks on
  modulation classifiers,'' \emph{IEEE Internet Things J.}, vol.~10, no.~21,
  pp. 19\,153--19\,162, 2023.

\bibitem{10402044}
Z.~Wang, W.~Liu, and H.-M. Wang, ``Wireless universal adversarial attack and
  defense for deep learning-based modulation classification,'' \emph{IEEE
  Commun. Lett.}, vol.~28, no.~3, pp. 582--586, 2024.

\bibitem{10416752}
A.~Ghasemi, E.~Zeraatkar, M.~Moradikia, and S.~Zekavat, ``Adversarial attacks
  on graph neural networks based spatial resource management in p2p wireless
  communications,'' \emph{IEEE Trans. Veh. Technol.}, vol.~73, no.~6, pp.
  8847--8863, 2024.

\bibitem{10742079}
B.~D. Son, N.~N. Khanh, T.~V. Chien, and D.~I. Kim, ``Adversarial attacks
  against double ris-assisted mimo systems-based autoencoder in
  finite-scattering environments,'' \emph{IEEE Wireless Commun. Lett.}, pp.
  1--1, 2024.

\bibitem{10830521}
J.~Shi, Q.~Zhang, W.~Zeng, S.~Li, and Z.~Qin, ``Secure transmission in wireless
  semantic communications with adversarial training,'' \emph{IEEE Commun.
  Lett.}, vol.~29, no.~3, pp. 487--491, 2025.

\bibitem{10436107}
W.~Zhang, M.~Krunz, and G.~Ditzler, ``Stealthy adversarial attacks on machine
  learning-based classifiers of wireless signals,'' \emph{IEEE Trans. Mach.
  Learn. Commun. Netw.}, vol.~2, pp. 261--279, 2024.

\bibitem{wu2021intelligent}
Q.~Wu, S.~Zhang, B.~Zheng, C.~You, and R.~Zhang, ``Intelligent reflecting
  surface-aided wireless communications: A tutorial,'' \emph{IEEE Trans.
  Commun.}, vol.~69, no.~5, pp. 3313--3351, 2021.

\bibitem{pan2020multicell}
C.~Pan, H.~Ren, K.~Wang, W.~Xu, M.~Elkashlan, A.~Nallanathan, and L.~Hanzo,
  ``Multicell mimo communications relying on intelligent reflecting surfaces,''
  \emph{IEEE Trans. Wireless Commun.}, vol.~19, no.~8, pp. 5218--5233, 2020.

\bibitem{zhang2021intelligent}
S.~Zhang and R.~Zhang, ``Intelligent reflecting surface aided multi-user
  communication: Capacity region and deployment strategy,'' \emph{IEEE Trans.
  Commun.}, vol.~69, no.~9, pp. 5790--5806, 2021.

\bibitem{mei2022intelligent}
W.~Mei, B.~Zheng, C.~You, and R.~Zhang, ``Intelligent reflecting surface-aided
  wireless networks: From single-reflection to multireflection design and
  optimization,'' \emph{Proc. IEEE}, vol. 110, no.~9, pp. 1380--1400, 2022.

\bibitem{yaswanth2023robust}
J.~Yaswanth, M.~Katwe, K.~Singh, S.~Prakriya, and C.~Pan, ``Robust beamforming
  design for active-ris aided mimo swipt communication system: A power
  minimization approach,'' \emph{IEEE Trans. Wireless Commun.}, vol.~23, no.~5,
  pp. 4767--4785, 2023.

\bibitem{papazafeiropoulos2023max}
A.~Papazafeiropoulos, P.~Kourtessis, and S.~Chatzinotas, ``Max-min sinr
  analysis of star-ris assisted massive mimo systems with hardware
  impairments,'' \emph{IEEE Trans. Wireless Commun.}, vol.~23, no.~5, pp.
  4255--4268, 2023.

\bibitem{zhang2022sum}
H.~Zhang, S.~Ma, Z.~Shi, X.~Zhao, and G.~Yang, ``Sum-rate maximization of
  ris-aided multi-user mimo systems with statistical csi,'' \emph{IEEE Trans.
  Wireless Commun.}, vol.~22, no.~7, pp. 4788--4801, 2022.

\bibitem{choi2024wmmse}
H.~Choi, A.~L. Swindlehurst, and J.~Choi, ``Wmmse-based rate maximization for
  ris-assisted mu-mimo systems,'' \emph{IEEE Trans. Commun.}, 2024.

\bibitem{zhang2025sum}
C.~Zhang, W.~U. Khan, A.~K. Bashir, A.~K. Dutta, A.~U. Rehman, and M.~M.
  Al~Dabel, ``Sum rate maximization for 6g beyond diagonal ris-assisted
  multi-cell transportation systems,'' \emph{IEEE Trans. Intell. Transp. Syst},
  2025.

\bibitem{9745781}
H.~Jiang, L.~Dai, M.~Hao, and R.~MacKenzie, ``End-to-end learning for ris-aided
  communication systems,'' \emph{IEEE Trans. Veh. Technol.}, vol.~71, no.~6,
  pp. 6778--6783, 2022.

\bibitem{10120965}
H.-Y. Chen, M.-H. Wu, T.-W. Yang, C.-W. Huang, and C.-F. Chou,
  ``Attention-aided autoencoder-based channel prediction for intelligent
  reflecting surface-assisted millimeter-wave communications,'' \emph{IEEE
  Trans. Green Commun. Netw.}, vol.~7, no.~4, pp. 1906--1919, 2023.

\bibitem{wu2025intelligent}
M.-H. Wu, H.-Y. Chen, T.-W. Yang, C.-C. Hsu, C.-W. Huang, and C.-F. Chou,
  ``Intelligent reflecting surface-assisted millimeter wave communications:
  Cross attention-aided variational autoencoder-based precoding design,''
  \emph{IEEE Trans. Cogn. Commun. Netw.}, 2025.

\bibitem{9977540}
H.~A. Le, T.~Van~Chien, V.~D. Nguven, and W.~Choi, ``Ris-assisted mimo
  communication systems: Model-based versus autoencoder approaches,'' in
  \emph{Proc. IEEE 33rd Annu. Int. Symp. Pers., Indoor Mobile Radio Commun.
  (PIMRC)}, 2022, pp. 1--6.

\bibitem{10136735}
H.~An~Le, T.~Van~Chien, V.~D. Nguyen, and W.~Choi, ``Double ris-assisted mimo
  systems over spatially correlated rician fading channels and finite
  scatterers,'' \emph{IEEE Trans. Commun.}, vol.~71, no.~8, pp. 4941--4956,
  2023.

\bibitem{10460991}
B.~D. Son, N.~T. Hoa, T.~V. Chien, W.~Khalid, M.~A. Ferrag, W.~Choi, and
  M.~Debbah, ``Adversarial attacks and defenses in 6g network-assisted iot
  systems,'' \emph{IEEE Internet Things J.}, vol.~11, no.~11, pp.
  19\,168--19\,187, 2024.

\bibitem{8811733}
Q.~Wu and R.~Zhang, ``Intelligent reflecting surface enhanced wireless network
  via joint active and passive beamforming,'' \emph{IEEE Trans. Wireless
  Commun.}, vol.~18, no.~11, pp. 5394--5409, 2019.

\bibitem{9531522}
T.~Van~Chien, H.~Q. Ngo, S.~Chatzinotas, B.~Ottersten, and M.~Debbah, ``Uplink
  power control in massive mimo with double scattering channels,'' \emph{IEEE
  Trans. Wireless Commun.}, vol.~21, no.~3, pp. 1989--2005, 2022.

\bibitem{mailloux2017phased}
R.~J. Mailloux, \emph{Phased array antenna handbook}.\hskip 1em plus 0.5em
  minus 0.4em\relax Artech house, 2017.

\bibitem{1175470}
D.~Gesbert, H.~Bolcskei, D.~Gore, and A.~Paulraj, ``Outdoor mimo wireless
  channels: models and performance prediction,'' \emph{IEEE Trans. Commun.},
  vol.~50, no.~12, pp. 1926--1934, 2002.

\bibitem{jiang2019turbo}
Y.~Jiang, H.~Kim, H.~Asnani, S.~Kannan, S.~Oh, and P.~Viswanath, ``Turbo
  autoencoder: Deep learning based channel codes for point-to-point
  communication channels,'' \emph{Proc. Adv. Neural Inf. Process. Syst.},
  vol.~32, 2019.

\bibitem{kiranyaz20211d}
S.~Kiranyaz, O.~Avci, O.~Abdeljaber, T.~Ince, M.~Gabbouj, and D.~J. Inman, ``1d
  convolutional neural networks and applications: A survey,'' \emph{Mech. Syst.
  Signal Process.}, vol. 151, p. 107398, 2021.

\bibitem{carlini2017towards}
N.~Carlini and D.~Wagner, ``Towards evaluating the robustness of neural
  networks,'' in \emph{Proc. IEEE Symp. Secur. Privacy (SP)}.\hskip 1em plus
  0.5em minus 0.4em\relax Ieee, 2017, pp. 39--57.

\bibitem{sutton2024adversarial}
O.~J. Sutton, Q.~Zhou, I.~Y. Tyukin, A.~N. Gorban, A.~Bastounis, and D.~J.
  Higham, ``How adversarial attacks can disrupt seemingly stable accurate
  classifiers,'' \emph{Neural Networks}, vol. 180, p. 106711, 2024.

\bibitem{goodfellow2014explaining}
I.~J. Goodfellow, J.~Shlens, and C.~Szegedy, ``Explaining and harnessing
  adversarial examples,'' \emph{arXiv preprint arXiv:1412.6572}, 2014.

\bibitem{3GPP_TR36_814}
``Further advancements for e-utra physical layer aspects,'' 3rd Generation
  Partnership Project (3GPP), Tech. Rep. TR 36.814 V9.0.0, March 2010, release
  9.

\bibitem{kingma2014adam}
D.~P. Kingma and J.~Ba, ``Adam: A method for stochastic optimization,''
  \emph{arXiv preprint arXiv:1412.6980}, 2014.

\bibitem{7790363}
T.~Van~Chien, E.~Björnson, and E.~G. Larsson, ``Multi-cell massive mimo
  performance with double scattering channels,'' in \emph{Proc. IEEE CAMAD},
  2016, pp. 231--236.

\end{thebibliography}

\end{document}